\documentclass{llncs}

\newcommand{\norm}[1]{}
\usepackage[utf8]{inputenc}
\usepackage{graphicx}
\usepackage{listings}
\usepackage{amssymb}
\usepackage{dsfont}
\usepackage{mathtools}
\usepackage{algorithm}
\usepackage{algpseudocode}
\usepackage{xcolor}
\usepackage{multirow}
\usepackage{amsmath}
\usepackage[english]{babel}
\usepackage{wasysym}
\usepackage{qip}
\usepackage{qm}
\usepackage{csquotes}
\usepackage{environ}

\usepackage[clock]{ifsym}
\usepackage{hyperref}
\usepackage[
	n,
	operators,
	advantage, 
	sets,
	adversary,
	landau,
	probability,
	notions,
	logic,
	ff,
	mm,
    primitives,
    events,
    complexity,
    asymptotics, 
    keys]{cryptocode}

\usepackage[maxbibnames=10]{biblatex}

\renewcommand{\qed}{\ensuremath{\blacksquare}}
\spnewtheorem{construct}{Construction}{\bfseries}{\itshape}
\spnewtheorem*{thm}{Theorem}{\bfseries}{\itshape}
\spnewtheorem*{lem}{Lemma}{\bfseries}{\itshape}

\newcommand{\hil}[1]{\ensuremath{\mathcal{#1}}}

\newcommand{\pos}[1]{\ensuremath{\mathrm{Pos}(#1)}}

\newcommand{\advs}{\ensuremath{\mathcal{A}_{\texttt{s}}}}
\newcommand{\ot}{\ensuremath{\mathcal{O}_{\texttt{s}}}}
\newcommand{\s}{\ensuremath{\texttt{s}}}
\newcommand{\otqot}[1]{\ensuremath{\textsf{BQS-OT}_{#1}}}

\begin{document}
\title{Powerful Primitives in the Bounded Quantum Storage Model}

\author{Mohammed Barhoush \and Louis Salvail}
\institute{Universit\'e de Montr\'eal (DIRO), Montr\'eal, Canada\\
\email{mohammed.barhoush@umontreal.ca}\ \ \ \email{salvail@iro.umontreal.ca}}
\maketitle

\begin{abstract}
The bounded quantum storage model aims to achieve security against computationally unbounded adversaries that are restricted only with respect to their quantum memories. In this work, we provide the following contributions in this model: 
\begin{enumerate}
 \item We build one-time programs and utilize them to construct CCA1-secure symmetric key encryption and message authentication codes. These schemes require no quantum memory from honest users, yet they provide information-theoretic security against adversaries with arbitrarily large quantum memories, as long as the transmission length is suitably large.
 
 \item We introduce the notion of $k$-time \emph{program broadcast} which is a form of program encryption that allows multiple users to each learn a single evaluation of the encrypted program, while preventing any one user from learning more than $k$ evaluations of the program. We build this primitive unconditionally and employ it to construct CCA1-secure asymmetric key encryption, encryption tokens, signatures, and signature tokens. All these schemes are information-theoretically secure against adversaries with roughly $e^{\sqrt{m}}$ quantum memory where $m$ is the quantum memory required for the honest user. 
\end{enumerate}
All of the constructions additionally satisfy disappearing security, essentially preventing an adversary from storing and using a transmission later on. 
\end{abstract}

  \setcounter{secnumdepth}{3}
   \setcounter{tocdepth}{3}
   \newpage
    \tableofcontents 
    \newpage
\section{Introduction}
Most of the interesting cryptographic concepts of security are unattainable when dealing with completely unrestricted adversaries. The conventional approach to resolve this conundrum is to focus solely on adversaries operating within polynomial time. Nevertheless, given our current understanding of complexity theory, security in this setting can only be guaranteed under the assumed hardness of solving
certain problems, such as factoring, computing the discrete log, or learning with errors. As a consequence, security in the \textit{computational model} is usually conditional.

An alternative approach is to constrain the quantum memory (qmemory) available to the adversary. This model is called the \textit{Bounded Quantum Storage Model} (\textsf{BQSM}) and was first introduced in 2005 by Damg{\aa}rd, Fehr, Salvail, and Schaffner \cite{dfss05,S07}. They demonstrated that non-interactive oblivious transfer and bit-commitment can be achieved information-theoretically in this model! Surprisingly, these schemes demand no qmemory from honest participants and can be made secure against adversaries with arbitrarily large qmemories by sufficiently extending the transmission length. Subsequently, this oblivious transfer scheme was later adapted to the more useful variant of non-interactive $1\textsf{-}2$ oblivious transfer\footnote{\textit{$1\text{-}2$ oblivious transfer} is a primitive allowing a sender to transmit two strings so that $(1)$ a receiver can choose to receive anyone of the two strings without learning anything about the other one and $(2)$ the sender is oblivious of the receiver's choice.} \cite{dfrss07}. Henceforth, we denote this $1\textsf{-}2$ oblivious transfer scheme as \textsf{BQS-OT}. 

Concretely, in the \textsf{BQSM}, an adversary $\mathcal{A}_{\texttt{s}}$ has access to unlimited resources at all times except at certain points. At these points, we say \emph{the memory bound applies}, and the adversary is forced to reduce its stored state to $\texttt{s}$-qubits. We emphasize that the adversary is again unrestricted with respect to its qmemory after the bound applies and is never restricted with respect to its computational power or classical memory.

In practice, the point when the memory bound applies can be implemented by pausing and delaying further transmissions which enforces the bound given the technological difficulties of maintaining a quantum state. Even in the foreseeable future, qmemory is not only expected to be very expensive but also very limited in size to allow for good enough fidelity. Consequently, this model appears to provide an apt characterization of real-world adversaries.

\subsection{Our Results} 

In this paper, we delve deeper into the \textsf{BQSM}. We first adapt conventional definitions, such as those pertaining to encryption, authentication, and one-time programs, to the \textsf{BQSM}. We then proceed to construct a variety of cryptographic primitives as elaborated on below.

We first build information-theoretic secure one-time programs \footnote{A one-time program (as introduced in \cite{gkr08}) $(1)$ can be used to obtain a single evaluation of the program on any input chosen by the user, and at the same time, $(2)$ (black-box obfuscation) cannot be used to obtain any other information about the program (see Definition \ref{def:BQS one-time}).} through utilizing the \textsf{BQS-OT} scheme. We then leverage one-time programs to construct information-theoretic secure CCA1-symmetric encryption and message-authentication codes. All these schemes are secure against any computationally unbounded adversary with $\s$ qmemory where $\s$ can be any fixed polynomial (in the security parameter) whereas honest users do not need any qmemory. 

Next, we address the challenges associated with information-theoretic asymmetric key cryptography, particularly the task of hiding the secret key from the public key without relying on computational assumptions. We propose the novel primitive termed program broadcast as a natural solution. 

Roughly speaking, a $k$-time program broadcast of a function $P$ allows arbitrarily many users to each $(1)$ obtain a single evaluation of the program, and, at the same time, $(2)$ cannot be used by any user to learn more than $k$ evaluations of the program. We proceed to construct an information-theoretically secure program broadcast and employ it to build CCA1-secure asymmetric key encryption, signatures, signature tokens, and encryption tokens\footnote{A signature token, as first defined in \cite{BSS16}, can be used to sign a single message (without the signing key) and then self-destructs. We similarly define decryption tokens -- each such token allows its holder to decrypt a single ciphertext and then self-destructs. }. 

The schemes built from program broadcast are secure against any computationally unbounded adversary with $\s$ qmemory where $\s$ can be any fixed polynomial in the security parameter but require $\lg^k{(\lambda)}$ qmemory for the receiver where $k\in \mathbb{R}$ can take any value larger than $2$. This implies that the gap between the qmemory required of the honest user and the needed qmemory to break security approaches $\lg^2{(\lambda)} \ \textsf{vs.}\ \poly[\lambda]$ which translates to a gap of $m \ \textsf{vs.} \ e^{\sqrt{m}}$ by setting $m= \lg^2{(\lambda)}$. While we cannot assert that the required qmemory for the honest receiver in the asymmetric setting is optimal, we show that it is not possible to achieve asymmetric key cryptography without requiring any qmemory for the honest user. Given that our qmemory requirements are already quite minimal, there is little room for improvement. For simplicity, in the rest of the paper we work with $k=3$, however, our results easily apply to any value $k>2$.

Prior to this work, none of the aforementioned primitives have been achieved in the \textsf{BQSM} and all are impossible information-theoretically in the plain model. 

Our constructions additionally satisfy disappearing security which shall be formally defined in Sec.~\ref{sec:dis and uncl} but we provide a brief overview here. A transmission, say a ciphertext or signature, is deemed disappearing if it can no longer be used after a certain point \cite{GZ21}. In the encryption setting, an adversary cannot decrypt a ciphertext even if the private key is revealed afterward. While in the authentication setting, an adversary cannot forge a signature of a message $\mu$ even if it has received a signature of $\mu$ earlier! Such disappearing properties are impossible in the plain model for obvious reasons, but in the \textsf{BQSM} an adversary cannot necessarily store a ciphertext or signature for later use.

\subsubsection{Notes on Feasibility.}

{The qmemory requirements for the receiver in our program broadcast, asymmetric key encryption, and signature schemes make these constructions difficult to actualize with current technology given the difficulty of storing quantum states. It is worth noting, however, that the users only need a small amount of qmemory and only need it at the start in order to process the public keys. Therefore, we hope that these schemes will become feasible with further advancements.}

{Meanwhile, our one-time programs, symmetric key encryption, and message authentication codes can be realized with current technology as they only require the ability to prepare, send, and measure a single qubit in the BB84 basis. Specifically, these constructions can be implemented on hardware designed to run quantum key distribution (QKD), albeit with a caveat. Indeed, popular \textit{weak coherent pulse sources} require interaction as the receiver needs to tell the sender which pulses were received. However, other technologies can enable non-interactive quantum transmissions to allow for non-interactive primitives such as one-time programs, encryption, and signature schemes. For instance, \textit{heralded photon sources} can, in principle, remove the need for interaction.

{Note also that our constructions can be modified to allow for users with imperfect apparatus to perform non-interactive error correction without compromising the security of the scheme as we briefly discuss in Sec.~\ref{sec:Noisy Communication}. With all that being said, in this work, we focus on building the theory and leave it as an open research direction to properly prepare the theoretical constructions for real-world applications.} 

\subsection{Related Work}
We compare our results with similar contributions in alternative models. 

\subsubsection{Computational Model.}

All of our contributions are impossible information-theoretically in the plain model. Even against polynomial-time adversaries, unconditionally secure encryption and authentication have remained an elusive goal. Signature tokens have only been built from heavy computational assumptions such as indistinguishability obfuscation \cite{cllz21}.

Meanwhile, one-time programs are impossible for the simple reason that any software can be stored and then reused to obtain multiple evaluations. Broadbent, Gutoski, and Stebila \cite{BGS13} extended this idea to the quantum realm and showed that quantum one-time programs are also impossible. These results do not apply to the \textsf{BQSM} as one-time programs are indeed feasible.

That being said, the works \cite{gkr08,GIS10} constructed one-time programs by leveraging one-time memory devices. These hardware devices essentially implement $1\textsf{-}2$ oblivious transfer -- each device holds two secrets and allows the device holder to choose and learn one and only one secret. Essentially, a circuit is encrypted and sent along with appropriate one-time memory devices. A receiver can only obtain enough keys from the devices to learn one evaluation from the encrypted circuit. Our one-time program construction involves replacing the one-time memory devices with the \textsf{BQS-OT} scheme, which serves the same purpose, although our proof requires novel analysis. 

\subsubsection{Bounded Classical Storage Model (\textsf{BCSM})}

Memory limitations were first considered in the classical setting. Specifically, in 1992, Maurer \cite{m92} introduced the \emph{bounded classical storage model} (\textsf{BCSM}) where adversaries are restricted only with respect to the size of their classical memory. A cipher was shown to guarantee privacy unconditionally, even against adversaries having access to a larger memory than what is needed for honest players. Subsequent works \cite{GZ19,DQW21,R18,DQW22} have achieved information-theoretic (interactive) oblivious transfer, key agreement, and symmetric key encryption and signatures in this model. Furthermore, Guan, Wichs, and Zhandry \cite{GWZ22} recently provided constructions for disappearing ciphertext and signature schemes based on public-key cryptography and one-way functions, respectively. 

Overall, the \textsf{BQSM} seems to provide stronger security guarantees when compared to the \textsf{BCSM} for several reasons. First of all, in the \textsf{BCSM}, the gap between the required memory to run many of these schemes securely and the size of the memory needed to break for various primitives, including digital signatures, is typically on the order of $m \ \text{vs.} \ m^2$, which is optimal \cite{DQW22}. This memory gap is a notable vulnerability as if it is feasible for honest users to store $m$ bits, it does not seem too difficult for a determined adversary to store $m^2$ bits. In contrast, the memory gap of $m \ \text{vs.} \ e^{\sqrt{m}}$ for primitives in the \textsf{BQSM} is far more significant. Moreover, certain primitives achievable in the \textsf{BQSM} are completely unattainable in the \textsf{BCSM}. For instance, non-interactive oblivious transfer was shown to be impossible in the \textsf{BCSM} \cite{DQW21} which implies the impossibility of the stronger one-time program primitive. Lastly, the \textsf{BCSM} does not provide the same unclonability guarantees given its classical nature, making encryption and signature tokens more challenging.

\subsubsection{Noisy Quantum Storage Model.}

In the \textit{Noisy Quantum Storage Model}, which was first introduced in \cite{wst08}, the qmemory of the adversary is not limited in size but is intrinsically noisy. The \textsf{BQS-OT} schemes proposed in \cite{dfss05,dfrss07} can, essentially as such, be shown secure against adversaries with noisy qmemories. In this work, we only study the \textsf{BQSM} and an interesting research direction is to generalize our results to the Noisy Quantum Storage Model. 

\section{Technical Overview}

We now discuss each of our results in more detail. We first describe how to construct one-time programs and then, how to use them to build symmetric key encryption. Next, we explain why we need a new tool to tackle asymmetric cryptography. Correspondingly, we introduce the notion of program broadcast and describe how to build it in the \textsf{BQSM}. We then show how program broadcast can be used to build asymmetric encryption, signatures, signature tokens, and encryption tokens. Finally, for more completeness, we provide two impossibility results related to the qmemory requirements and disappearing properties of the schemes. 

\subsection{One-Time Programs}
\label{sec:one-time programs intro}
We now present a quantum algorithm $\mathcal{O}$ that converts any polynomial classical circuit to a one-time program in the \textsf{BQSM} (Theorem \ref{main poly}). Fortunately, the \textsf{BQSM} dodges all the impossibility proofs of one-time programs. The work showing the impossibility of quantum one-time programs \cite{BGS13} does not apply to our setting since our programs cannot be stored by the adversary to be reused. Moreover, the proof of the impossibility of quantum obfuscation \cite{BGS13} does not apply since it requires the adversary to run the quantum circuit homomorphically. Adversaries in our setting are forced to continuously perform measurements on quantum states (due to the memory bound) which interrupts a quantum homomorphic evaluation. 

Before describing the construction, we discuss its intuition. A natural first attempt is to apply previous one-time program constructions \cite{gkr08,GIS10} based on one-time memory devices but replace the devices with \textsf{BQS-OT} transmissions. The problem is that simulating an adversary requires determining which 1-out-of-2 strings are received from the oblivious transfer transmissions. This information allows the simulator to deduce which evaluation the adversary has learned from the program. In earlier constructions, this is not a problem since the adversary must \emph{explicitly} ask which string it wishes to receive from the one-time memory devices. Note that Kilian \cite{k88} builds a similar notion known as \emph{oblivious circuit evaluation} from OT but the simulator is also assumed to be given this information. In these works, the simulator can simply run the adversary and `know' which strings are requested. However, in our case, it is not clear this extraction is always possible -- a complex adversary might not know itself which 1-out-of-2 strings it has received! To make matters worse, our adversary is a quantum and computationally unbounded algorithm. Note that the proof of security for \textsf{BQS-OT} does not provide a method to extract the receiver's chosen bit but rather just shows that such a bit must exist.

To solve this issue, we construct a simulator that (inefficiently) analyzes the circuit of the adversary gate-by-gate to extract the chosen bit. In particular, the simulator takes the circuit of the adversary and {classically simulates it on all possible \textsf{BQS-OT} transmissions
}. This allows the simulator to approximately determine the distribution of the measurement results obtained by the adversary given the transmission provided to the adversary. This then allows the simulator to determine the distribution of the \textsf{BQS-OT} transmission from the perspective of the adversary given the measurement results obtained. In other words, the simulator can infer which part of the transmission the adversary is uncertain about which can be used to determine which string the adversary is oblivious to. This technical extraction result is given in Lemma \ref{determin C} and may be of independent interest. 

The rest of the construction involves adapting Kilian's approach \cite{k88} for building one-time programs from oblivious transfer to the quantum realm. This part is not a main contribution of this work, so is not discussed further here. Readers are referred to Sec.~\ref{sec:one-time short} for the formal statement of the result. However, the full construction is given in App.~\ref{sec:one-time construction} as it is quite long and detailed. 

{
An evaluator must choose its input and process the quantum part of our one-time program during the transmission. In particular, the input determines the measurement bases used in the oblivious transfer transmissions. This does not prevent the evaluator from storing the resulting classical data, or the output, for later use. However, it does prevent the evaluator from receiving the entire program first and choosing an arbitrary input after the transmission ends. For general circuits, this limitation is inherent when the honest evaluator has no qmemory. A classical state that allows the evaluator to choose any input later could be copied and used to evaluate the program on several inputs. 
}

\subsection{Symmetric Key Encryption}

We use one-time programs to build a symmetric key encryption scheme that satisfies indistinguishability under (lunchtime) chosen ciphertext attack (IND-CCA1) \cite{CS98} and with disappearing ciphertexts. IND-CCA1 tests security against an adversary that can query the encryption oracle at any point in the experiment and can query the decryption oracle at any point before the challenge ciphertext is received. To satisfy \textit{disappearing} security under CCA1, we require that such an adversary cannot win the IND-CCA1 experiment even if the private key is revealed after the challenge ciphertext is given.

To achieve IND-CPA (where the adversary is not given access to a decryption oracle), the ciphertext can simply consist of a one-time program that evaluates to $\perp $ on every input except on the private key, where it outputs the message. By the security of one-time programs, an adversary has a negligible chance of guessing the key and learning the message. Upgrading to CCA1-security turns out to be far more involved as we now discuss.

First of all, it is trivial to show that we must rely on quantum ciphertexts for unconditional security in the \textsf{BQSM}. Now, traditionally, a CPA-secure scheme can be upgraded to a CCA-secure one by simply authenticating ciphertexts using message-authentication codes or a signature scheme. However, it seems difficult to authenticate a quantum ciphertext in a way that allows verification without any qmemory. Alternatively, we could authenticate the output of the one-time programs in the CPA construction. Specifically, instead of sending programs that map the secret key to the message, we send programs that map the key to the message along with a tag or signature of the message. 

This solution fails due to the following attack: an adversary queries the encryption oracle on an arbitrary message $\mu$, receives a one-time program, modifies the one-time program so that it outputs $\perp$ on any input starting with 0, and forwards the modified one-time program to the decryption oracle. Depending on the response received from the oracle, the adversary can successfully determine the first bit of the secret key. This attack can be repeated a polynomial number of times to deduce the entire secret key. It turns out an adversary can perform a variety of attacks of this sort due to a fundamental problem with the \textsf{BQSM}: this model provides security by ensuring that users only receive partial information from a transmission. Hence, it is difficult to detect whether an adversary has tampered with a transmission. 

To thwart such attacks, we rely on a construction where a ciphertext consists of two interdependent one-time programs. The first program initiates some randomness and this randomness is used to ensure that the second one-time program is evaluated on a seemingly random input during decryption. To prevent the adversary from choosing the first one-time program, we authenticate the output of the first program in the second program using the secret key. We show that it is difficult to modify both programs simultaneously in a meaningful way without being detected. Proving this rigorously is somewhat technical as the adversary can still perform various modifications successfully -- the reader is referred to Theorem \ref{construction 1} for the formal proof. 

\subsection{The Obstacle to Asymmetric Cryptography}

Asymmetric cryptography utilizes a pair of related keys, a secret and a public key, to enable versatile cryptographic applications. Specifically, our goal is to build information-theoretic secure asymmetric key encryption and digital signatures in the \textsf{BQSM}. However, the task of concealing the secret key from the public key without relying on computational assumptions poses a significant challenge. Such assumptions should be avoided in the \textsf{BQSM} given the goal of this model is to base security purely on the memory limitations of the players. 

This implies we need to impose the memory bound in some way during the public key transmission in order to hide the secret key. In the classical bounded storage model, this can be done by announcing a large classical public key \cite{DQW22}. However, in our scenario, imposing the memory bound necessitates the use of quantum public keys. Unfortunately, due to no-cloning, we need to create and distribute multiple copies of the quantum key -- this is the general approach taken in the computational setting as well \cite{GC01,D21,KKN05}. Here we face a critical problem: if a computationally unbounded adversary gains access to multiple copies of the key, it can repeatedly reuse its quantum memory to process multiple keys. Such an adversary could gradually extract classical information from each key copy until it has learned the entire quantum public key.

To prevent this attack, it is imperative that every public key copy needs \emph{additional} qmemory to process, preventing an adversary from processing an unlimited number of keys. In other words, we need to distribute keys in a way so that all users must process the keys simultaneously or in parallel. More abstractly, we want to distribute quantum information in a way that allows every user to learn some information while preventing a computationally unbounded adversary from learning too much information. This goal is a recurring theme in various settings, hence, we introduce a new notion which we term \emph{program broadcast} that formalizes and captures this objective. We build this primitive unconditionally in the \textsf{BQSM} and employ it to build asymmetric cryptography.

All in all, this is the first work to achieve information-theoretic asymmetric key encryption in any of the bounded (classical or quantum) storage models. This is the main focus and contribution of this paper. In the following, we first discuss program broadcast, which may be of independent interest, and then how this can be used to realize asymmetric cryptography and even encryption or signature tokens. 

\subsection{Program Broadcast}

In many situations, we would like to send multiple one-time programs to multiple users while ensuring that no adversary can take advantage of all the information to learn the program. This is the goal of {program broadcast}. Recall that a $k$-time program broadcast of a function $P$ allows an unbounded polynomial number of users to each $(1)$ obtain at least a single evaluation of the program, and, at the same time, $(2)$ cannot be used by any user to learn more than $k$ evaluations of the program. Essentially, an adversary with access to the broadcast can be simulated with access to $k$ queries to an oracle for $P$. 

In Sec.~\ref{sec:broadcast}, we present a scheme for an information-theoretically secure program broadcast in the \textsf{BQSM}. The idea is to distribute a polynomial number of augmented one-time programs of $P$ in a fixed time period. Unlike the one-time programs discussed in the previous section, these augmented versions require a small amount of qmemory to process. 

More rigorously, there is a set time where the sender distributes one-time programs of the form $\mathcal{O}(P\oplus c_i)$ where $\mathcal{O}$ is our one-time compiler. In each one-time program, the output of $P$ is padded with a different value $c_i$. A quantum ciphertext is generated which encrypts the value $c_i$ and is sent with the corresponding one-time program. Users can evaluate the one-time program but need to store the ciphertext states until the encryption key is revealed to make use of their evaluation. The key is revealed at a set time after all users have received a one-time program copy. By revealing this classical key at the end, we are essentially forcing users to process the programs in a parallel fashion, limiting the number of evaluations that can be learned. In other words, an adversary with bounded qmemory can store the ciphertext states for a limited number of one-time programs and thus can only learn a limited number of evaluations of $P$. 

To prove this rigorously, we need a new min-entropy lower bound (Lemma \ref{splitting cor}) which roughly states that if $H_{\infty}(X_0X_1...X_{p-1}|Q)\geq \alpha$, where $X_i\in \{0,1\}^n$ and $Q$ is some quantum information, then there exists $i\in [p]$ and negligible $\epsilon$ such that roughly $H^\epsilon_{\infty}(X_i)\geq \alpha/p$. Konig and Renner \cite{KR11} established such a bound but only in the case when $n$ is sufficiently large with respect to $p$. Unfortunately, in our applications, $n$ is very small with respect to $p$. Informally, we resolve this issue by using their bound when $p$ is small and then use this as a base case to an inductive argument for larger values of $p$ achieving the required result. This result is of independent interest in the field of information theory. 

Note that while one-time programs can be built from one-time memory devices in the plain model, this is insufficient to build program broadcast in the plain model. Specifically, we need to use the qmemory bound outside the one-time programs as well in order to construct this primitive. Hence, program broadcast acts as a more pronounced advantage of the \textsf{BQSM}. 

\subsection{Asymmetric Key Encryption}

We now discuss how we upgrade our symmetric key encryption scheme to a CCA1-secure asymmetric key scheme using program broadcast. We let the private key be a large program $P$ while the public key is a quantum program broadcast of $P$. A user can use the broadcast to learn a single evaluation of $P$ which can be used to encrypt messages in the same way as in the private key setting. However, an adversary can only learn a limited number of evaluations by the security of the program broadcast. If $P$ is large enough then this knowledge is not sufficient to predict the output of $P$ on a random input. In other words, the adversary cannot predict the encryption key of another user, so security is reduced to the symmetric key setting. 

It is not difficult to show that the public keys need to be mixed-state for information-theoretic security in the \textsf{BQSM}. This might seem unfortunate, but we believe that it is not very relevant whether keys are pure or mixed-state. Some works, such as \cite{BMW23}, require that the quantum public keys be pure since this would provide some level of certification by allowing users to compare keys with a SWAP test. However, this test is not always feasible to perform in the \textsf{BQSM} given that keys cannot necessarily be stored. Instead, we provide a more secure way to certify mixed-state quantum keys without establishing authenticated quantum channels in a companion work \cite{BS232}. That being said, in this work, we do not certify public keys and it is always assumed that honest users receive authentic quantum public keys. This is a common assumption for the quantum public key schemes, such as those in \cite{GC01,D21,KKN05}, as a quantum state cannot be signed \cite{BCG02}. 

\subsection{Encryption Tokens}

Next, we introduce a novel primitive, which we term encryption and decryption tokens. A decryption token can be used to decrypt a single ciphertext, while an encryption token can be used to encrypt a single message. We present an information-theoretically secure token scheme with an unbounded polynomial number of encryption and decryption tokens in the \textsf{BQSM} as a direct application of program broadcast. Fundamentally, the underlying encryption scheme is just a $q$-time pad which is a generalization of the one-time pad but that allows for the encryption of $q$ messages. This gives a simple solution where the encryption and decryption tokens are one-time programs of the encryption and decryption functions. This is not very satisfactory as it only allows for the distribution of $q$ tokens.

To allow for an unbounded polynomial number of tokens, we broadcast the encryption and decryption functionalities of the $q$-time pad. However, our program broadcast only lasts for a fixed duration and, unlike our public keys in the asymmetric key encryption scheme, users in this setting may request multiple encryption or decryption tokens at varying times. Hence multiple `waves' of token broadcasts are sent. In each wave, the adversary can process a limited number of tokens by the security of the program broadcast. A new key for the $q$-time pad is initiated for each broadcast and hence the private key for this construction grows linearly with time. We say the tokens have a \emph{sparsity} restraint since a limited number can be processed during the broadcast. In fact, this restraint can be viewed as an advantage in some scenarios; similar delay restraints were studied in \cite{AGK20} but in the authentication setting. We also introduce order restraints, allowing an authority to control the order in which users perform encryption or decryption operations. This is easily incorporated into the encryption and decryption functionalities. See App.~\ref{sec:enc tok} for more details. 

\subsection{Signatures and Signature Tokens}

We then turn to authentication and construct a signature scheme for classical messages that is information-theoretically unforgeable under chosen message and verification attacks in the \textsf{BQSM}. We follow this with the construction of signature and verification tokens and imbue these tokens with sparsity and order restraints. Note that the signature and token schemes differ from traditional constructions in the standard model where the verification key is classically announced. In the \textsf{BQSM}, we cannot ensure security if the verification key is classical so we must distribute quantum copies. All these constructions are very similar to the corresponding encryption constructions and hence are not discussed further. The reader is referred to Sec.~\ref{sec:authentication} for more details. 


\subsection{Impossibility Results}

We present two interesting impossibility results in this work. First, in Sec.~\ref{sec:discussion}, we show that it is impossible to implement asymmetric schemes (with unbounded public key copies) in the \textsf{BQSM} without requiring qmemory from the honest user. The fundamental idea is that if the public keys require no qmemory to process then an adversary can request and process an unbounded number of key copies which would reveal information about the secret key. The technical argument is more involved since the public keys may be mixed-state which could aid in hiding information.

The second impossibility result is concerned with the disappearing property of our constructions. The ciphertexts, signatures, and tokens in our schemes satisfy disappearing security providing interesting applications as we discussed earlier. However, this disappearing property can also be disadvantageous in some scenarios. For instance, sometimes, it is more beneficial to have an obfuscation that can be stored and reused multiple times instead of a disappearing one-time program. The final contribution of this paper is to provide a negative result showing that non-disappearing obfuscation and one-time programs are impossible in the \textsf{BQSM} (see Sec.~\ref{sec:impossible}). 

The proof relies on a new approach since our simulator is computationally unbounded, which nullifies adversarial attacks used in standard impossibility proofs \cite{BGI+01,GK05,BCC14,BGS13}. Essentially, we show that if a computationally unbounded adversary can perform a single evaluation after the qmemory bound applies then it can rewind the program and perform another evaluation. By the gentle measurement lemma \cite{W99}, this rewinding introduces only a negligible error each time which allows the adversary to learn a super-polynomial number of evaluations with non-negligible probability. These evaluations can be used to learn more information regarding the hidden program than any simulator can learn using only a polynomial number of oracle queries. 

\section{Preliminaries}
\label{sec:Preliminaries}

\subsection{Notation}

In the following, we denote Hilbert spaces by calligraphic letters: $\hil{H}$, $\hil{V}$, {
etc.
}
The set of positive semi-definite operators in Hilbert space $\hil{H}$ are denoted $\pos{\hil{H}}$. When Hilbert space $\hil{X}$ holds a classical value, we denote the random variable for this value by $X$.

We denote the density matrix of a quantum state in a register $E$ as $\rho_E$. If a quantum state depends on a classical variable $X$, then we denote $\rho^x_E$ as the density matrix of the state given $X=x$. Meanwhile, to an observer who does not know the value of $X$, the state is given by $\rho_E\coloneqq \sum_{x}P(X=x)\rho_E^x$ and the joint state by $\rho_{XE}\coloneqq \sum_{x}P(X=x) \proj{x}\otimes \rho_E^x$. Such states having a quantum part depending on a classical value are called \emph{cq-states}. The random variable $X$
then corresponds to the classical state $\rho_X\coloneqq \sum_x P(X=x)\proj{x}$. 

We denote the trace distance between two density matrices as $\delta(\rho,\sigma)\coloneqq \frac{1}{2}\trace{\lvert \rho-\sigma\rvert}$ and the trace norm as $\|\rho\|_1\coloneqq \trace{\sqrt{\rho \rho^{*}}}$. 

Notice that $\rho_{XE}=\rho_X\otimes \rho_E$ if and only if $\rho_E $ is independent of $X$ (i.e. $\rho_E^x=\rho_E$) which means that no information can be learned regarding $\rho_E$ from $X$ and vice versa. More generally, if $\rho_{XE}$ is $\epsilon$-close to $\rho_E\otimes \rho_X$ in terms of trace distance denoted as $\delta(\rho_{XE},\rho_X\otimes \rho_E)\leq \epsilon$ or equivalently $\rho_{XE}\approx_{\epsilon}\rho_{X}\otimes \rho_E$, then no observer can distinguish between the two systems with advantage greater than $\epsilon$. We write $\langle \rho_A,\rho_B\rangle$ to denote a sequential transmission where register $A$ is sent and then register $B$.

The rectilinear or $+$ basis for the complex space $\mathbb{C}^2$ is the pair $\{|0\rangle,|1\rangle\}$ while the diagonal or $\times$ basis is the pair $\{|0\rangle_{\times},|1\rangle_{\times}\}$ where $|0\rangle_{\times}\coloneqq \frac{|0\rangle +|1\rangle}{\sqrt{2}}$ and $|1\rangle_{\times}\coloneqq \frac{|0\rangle -|1\rangle}{\sqrt{2}}$. To choose between the $+$ and $\times$ basis depending on a bit $b\in\{0,1\}$, we write $\{+,\times\}_b$. 

We say $x\in_R X$ if $x$ is chosen uniformly at random from the values in $X$. We let $\mathds{1}_k$ denote the density matrix of a uniformly distributed variable on $k$ elements.
Also, let $[n]\coloneqq [0,1,\ldots,n-1]$ and let $\negl[n]$ be denoting any function that is smaller than the inverse of any polynomial for large enough $n$. 

We say an algorithm $A$ is \emph{QPT} if it is a quantum polynomial-time algorithm. We let $A^P$ denote an algorithm with access to a polynomial number of classical oracle queries to the program $P$ and let $A^{qP}$ denote access to $q$ classical oracle queries.  

Recall a Pauli pad \cite{MTW00} information-theoretically hides a $n$-qubit state using a $2n$-bit key $k$. We denote $\textsf{PP}_k(\rho)$ to be the Pauli pad of the state with density matrix $\rho$ using key $k$.

By abuse of notation, if $x$ is a binary string and $M$ is a matrix then, for simplification, we let $M\cdot x$ denote the string representation of the vector $M\cdot \vec{x}$. 

\subsection{Rényi Entropy}

We remind the reader of the notion of quantum Rényi Entropy and its variants. See \cite{KRS09} for more detailed definitions.  

\begin{definition}[Conditional Min-Entropy]
\label{def:cond entropy}
Let $\rho_{XB}\in \pos{\hil{H}_X\otimes \hil{H}_B}$ be classical on $\hil{H}_X$. The min-entropy of $\rho_{XB}$ given $\hil{H}_B$ is defined as
\begin{align*}
  H_{\infty}(X|B)_{\rho} \coloneqq - \lg{\left(p_{\textnormal{guess}}(X|B)_\rho\right)}
\end{align*}
where $p_{\textnormal{guess}}(X|B)_\rho$ is the maximal probability to decode $X$ from $B$ with a POVM on $\hil{H}_B$. 
\end{definition}

\begin{definition}[Smooth Min-Entropy]
Let $\epsilon \geq 0$ and $\rho_{XB}\in \pos{\hil{H}_X\otimes \hil{H}_B}$. The $\epsilon$-smooth min-entropy of $\rho_{XB}$ given $\hil{H}_B$ is defined as 
\begin{align*}
  H_{\infty}^{\epsilon}(X|B)_{\rho}\coloneqq \sup_{\overline{\rho}}H_{\infty}(X|B)_{\overline{\rho}}
\end{align*}
where the supremum is taken over all density operators $\overline{\rho}_{XB}$ acting on $\hil{H}_X\otimes \hil{H}_B$ such that $\delta(\overline{\rho}_{XB},\rho_{XB})\leq \epsilon$. 
\end{definition}

\subsection{Uncertainty Relations}
\label{sec:uncertainity}

This section discusses lower bounds on the average min-entropy and presents a new generalized min-entropy splitting lemma conditioned on quantum states. This will be fundamental in our program broadcast construction. 

The min-entropy satisfies the following chain rule.

\begin{lemma}[Chain Rule for Min-Entropy \cite{RK04}]
\label{chain}
Let $\epsilon\geq 0$ and $\rho_{XUE}\in \pos{\hil{H}_X\otimes \hil{H}_U\otimes \hil{H}_E}$ where register $E$ has size $m$. Then,
\begin{align*}
  H_{\infty}^{\epsilon}({X}|UE)_{\rho} \geq H^{\epsilon}_{\infty}(XE|U)_{\rho}-m .
  \end{align*}
\end{lemma}

\begin{lemma}[Uncertainty Relation \cite{dfss05}]
\label{uncertinity relation}
Let $\rho \in \mathcal{P}(\textsf{H}_2^{\otimes n})$ be an arbitrary $n$-qubit quantum state. Let $\Theta \in_R \{+,\times \}^n$ and let $X$ be the outcome when $\rho$ is measured in basis $\Theta$. Then for any $0< \gamma <\frac{1}{2}$,

\begin{align*}
  H^{\epsilon}_{\infty}(X|\Theta)_{\rho}\geq \left(\frac{1}{2}-2\gamma \right)n ,
\end{align*}
where $\epsilon=\exp(-\frac{\gamma^2n}{32(2-\lg(\gamma))^2})$.
\end{lemma}

If $X=X_0X_1$ has high entropy, then it was shown in \cite{dfrss07} that, in a randomized sense, one of $X_0$ or $X_1$ has high entropy. This result was crucial in the construction of 1-out-of-2 oblivious transfer in the \textsf{BQSM} (Construction \ref{con:QOT}). 

\begin{lemma}[Conditional Min-Entropy Splitting Lemma \cite{dfrss07}]
\label{splitting first}
Let $\epsilon\geq 0$ and let $X_0,X_1$ and $Z$ be random variables. If $H_{\infty}^{\epsilon}(X_0X_1|Z)\geq \alpha$, then there exists a binary random variable $C$ such that $H_{\infty}^{\epsilon+\epsilon'}(X_{1-C}|ZC)\geq \alpha/2-1-\log(1/\epsilon')$ for any $\epsilon'>0$. 
\end{lemma}



More generally, min-entropy sampling gives lower bounds on the min-entropy of a subset of a string given the min-entropy of the entire string. In \cite{KR11}, it was shown how to sample min-entropy relative to quantum knowledge which allows us to give a lower bound on min-entropy conditioned on a quantum state.

\begin{lemma}[Quantum Conditional Min-Entropy Splitting Lemma \cite{KR11}]
\label{splitting original}
Let $\epsilon\geq 0$ and let $\rho_{X_0X_1B}$ be a quantum ccq-state where $X_0$ and $X_1$ are classical random variables on alphabet $\mathcal{X}$ of dimension $d=\lg \lvert \mathcal{X}\rvert>14$. Then, there exists a binary random variable $C$ such that for all $\tau \geq 0$, 
\begin{align*}
  H_{\infty}^{\epsilon+\tau}(X_{C}|CB)_{\rho}\geq \frac{H_{\infty}^{\tau}(X_0X_1|B)_{\rho}}{2}-2
\end{align*}
where $\epsilon \coloneqq 2^{-\frac{3d}{2}+1}$.
\end{lemma}

\begin{proof}
Notice that any distribution $P_C$ over $\{0,1\}$ is a $(2,3/4,0)$-sampler (Definition 2.1 \cite{KR11}). Corollary 6.19 in \cite{KR11} gives the result. 
 \qed     \end{proof}

In Lemma \ref{splitting original} the size of the sampled string $X_C$ is half the size of the original string $X$. The results in \cite{KR11} do not give strong lower bounds when the sampled string is small relative to the size of the original string. Hence, we present the following generalization of Lemma \ref{splitting original} which will be used in constructing program broadcast and asymmetric key encryption. 

\begin{lemma}[Generalized Min-Entropy Splitting Lemma]
\label{splitting cor}
Let $\epsilon\geq 0$. Let $\rho_{XB}$ be a quantum cq-state where $X\coloneqq X_0X_1...X_{\ell}$ and each $X_i$ is a classical random variable over alphabet $\mathcal{X}$ of dimension $d>14$. There exists a random variable $C$ (with $\lceil \lg \ell\rceil$ bits) such that for all $\tau \geq 0$, \begin{align*}
  H_{\infty}^{\epsilon+\tau}(X_{C}|CB)_{\rho}\geq \frac{H_{\infty}^{\tau}(X|B)_{\rho}}{\ell}-4
\end{align*}
where $\epsilon\coloneqq 2^{-\frac{3d}{2}+2}$.
\end{lemma}

\begin{proof}
We only consider the case $\ell= 2^k$, however, the same argument works for the general case. We prove the following slightly stronger bound using induction on $k$:
\[
H_{\infty}^{\epsilon_k+\tau}(X_{C}|CB)\geq \frac{H_{\infty}^{\tau}(X|B)}{2^k}-4(1-\frac{1}{2^{k}}) .
\]
where $\epsilon_k\coloneqq \sum_{i=0}^{k-1}2^{-2^id+1}$. The base case $k=1$ follows trivially from the Quantum Conditional Min-Entropy Splitting Lemma \ref{splitting original}. Assume the claim holds for $k=n$. For $k=n+1$, we apply Lemma \ref{splitting original} on the two variables $Y_0\coloneqq X_1...X_{2^n}$ and $Y_1\coloneqq X_{2^n+1}...X_{2^{n+1}}$ to deduce that there exists a binary variable $C_1$ such that $H_{\infty}^{\epsilon'+\tau}(Y_{C_1}|BC_1)\geq \frac{H_{\infty}^{\tau}(X|B)}{2}-2$ where $\epsilon'\coloneqq 2^{-{2^nd}+1}$. By the inductive hypothesis, there exists $C'$ such that: \begin{align*}
  H_{\infty}^{\epsilon_{n}+\epsilon'+\tau}(X_{C_1C'}|BC_1C')&\geq \frac{H_{\infty}^{\epsilon'+\tau}(Y_{C_1}|BC_1)}{2^{n}}-4\left(1-\frac{1}{2^{n}}\right)\geq\\
 \frac{\frac{H_{\infty}^{\tau}(X|B)}{2}-2}{2^{n}}-4\left(1-\frac{1}{2^{n}}\right)
  &=\frac{H_{\infty}^{\tau}(X|B)}{2^{n+1}}-4\left(1-\frac{1}{2^{n+1}}\right) .
\end{align*} 

\qed     \end{proof}

\subsection{Privacy Amplification}
\label{sec:privacy amplification}

For the rest of this work, we let $\textsf{H}_{m,\ell}$ denote a two-universal class of hash functions from $\{0,1\}^m$ to $\{0,1\}^{\ell}$~\cite{CW79}.
Remember that a class of hash functions is \emph{two-universal} if for every distinct $x, x'\in\{0,1\}^m$ 
and for $F \in_R \textsf{H}_{m,\ell}$, we have $Pr[F(x)=F(x')]\leq \frac{1}{2^{\ell}}$. 

\begin{theorem}[Privacy Amplification \cite{RK04}]
\label{privacy amplification}
Let $\epsilon \geq 0$. Let $\rho_{XB}\in \pos{\hil{H}_X\otimes \hil{H}_{B}}$ be a cq-state 
where $X$ is classical and takes values in $\{0,1\}^m$. 
Let $F\in_R \textsf{H}_{m,\ell}$ be the random variable for a function chosen uniformly at random in a two-universal class
of hash functions $\textsf{H}_{m,\ell}$. Then, 
\begin{align*}
  \delta(\rho_{F(X)FB},\mathds{1}\otimes \rho_{FB})\leq \frac{1}{2}2^{-\frac{1}{2}(H^{\epsilon}_{\infty}({X}|B)_{\rho}-\ell)}+\epsilon .
\end{align*}
\end{theorem}

\subsection{Matrix Branching Programs}
\label{sec:branching}

Our one-time compiler is first constructed for a class known as matrix branching programs which are introduced in this section. Let $\hil{M}_k$ denote the set of binary permutation matrices of dimension $k$. 

\begin{definition}[Matrix Branching Program]
\label{def:mbp}
A \emph{$(Q_{rej}, Q_{acc})$-matrix branching program} of length $N$, width $k$ and input of length $n$ is given by $\mathcal{P}=(M_j^b)_{j\in [N],b\in\{0,1\}}$ where $M^b_j\in \hil{M}_k$. For any input $w\in \{0,1\}^n$, $\mathcal{P}$ performs the following function:
\begin{align*} 
\begin{split}
\mathcal{P}(w) = \begin{cases} 
0 & \prod_{j=0}^{N-1}M_j^{w_{(j \bmod n)}}= Q_{rej}\\
1 & \prod_{j=0}^{N-1}M_j^{w_{(j \bmod n)}}= Q_{acc}\\ 
\texttt{undef} & \textnormal{otherwise}.\\
\end{cases}
\end{split}
\end{align*}
For any input $w$, the set of matrices $(M_j^{w_{(j \bmod n)}})_{j\in [N]}$ which are used during the evaluation $\mathcal{P}(w)$ are referred to as a \emph{consistent path} in the program. 
\end{definition}

Note that the definition assumes that input bits are requested in an ordered manner where the $j^{th}$ instruction requests bit $w_{(j\bmod n)}$. This deviates from the standard definition of branching programs where the input bits may be requested in an arbitrary manner. However, a standard program can be easily transformed into the form described by adding redundant identity matrices which potentially blows up its size by a factor of $n$. It is necessary to put the programs in this form as otherwise, the adversary can easily differentiate between programs by inspecting the order in which input bits are requested. 

These branching programs cover a wide range of functions. In fact, for any size $n$ and depth $d$, there exist permutations $Q_{acc}$ and $Q_{rej}$ such that any $\textbf{NC}^1$ circuit of size $n$ and depth $d$ can be converted to a $(Q_{rej},Q_{acc})$-matrix branching program of width $k=5$ and length $4^d$ by Barrington's Theorem \cite{B89}. 

In our case, a circuit of depth $d$ can be converted to a branching program of length $4^dn$ where the factor $n$ is the result of the potential blow-up from requiring the input access to be ordered. Without loss of generality, in the rest of the work, it is assumed that all branching programs have the same accepting and rejecting matrices and that the programs do not have \textit{undefined} outputs on valid inputs.

\subsection{\texorpdfstring{$1\textsf{-}2$}{1-2} Oblivious Transfer}
\label{sec:OT}

In 1\textsf{-}2 oblivious transfer, a sender $\mathsf{S}$ sends two strings to the receiver $\mathsf{R}$ so that $\mathsf{R}$ can choose to receive one of the strings but does not learn anything about the other and $\mathsf{S}$ is oblivious of $\mathsf{R}$'s choice. 

Let $S_0$ and $S_1$ be two $\ell$-bit random variables representing the strings $\mathsf{S}$ sends and let $C$ denote a binary random variable describing $\mathsf{R}$'s choice. Also, let $Y$ denote the $\ell$-bit random variable $\mathsf{R}$ outputs which is supposed to be $Y=S_C$. 

\begin{definition}[$1\textsf{-}2\ \textsf{QOT}$]
An $\epsilon$-secure $1\textsf{-}2$ \emph{quantum oblivious transfer} (QOT) is a quantum protocol between a sender $\mathsf{S}$ and receiver $\mathsf{R}$ with input $C\in \{0,1\}$ and output $Y$ such that: 
\begin{itemize}
\item For honest $\mathsf{S}$ and $\mathsf{{R}}$, and any distribution of $C$, $\mathsf{R}$ can get $Y=S_C$ except with probability $\epsilon$.
\item $\epsilon$-Receiver-security: If $\mathsf{R}$ is honest then for any sender $\hat{\mathsf{S}}$, there exists variables $S_0',S_1'$ such that $Pr[Y=S_C']\geq 1 - \epsilon$ and 
\begin{align*}\delta (\rho_{CS_0'S_1'\hat{\mathsf{S}}},\rho_C\otimes \rho_{S_0'S_1'\hat{\mathsf{S}}})\leq \epsilon .\end{align*}
\item $\epsilon$-Sender-security: If $\mathsf{S}$ is honest, then for any $\hat{\mathsf{R}}$, there exists a binary random variable $C'$ such that 
\begin{align*}\delta (\rho_{S_{1\textsf{-}C'}S_{C'}C'\hat{\mathsf{R}}},\rho_{S_{1\textsf{-}C}}\otimes \rho_{S_{C'}C'\hat{\mathsf{R}}})\leq \epsilon
.
\end{align*}
\end{itemize}
\end{definition}

Informally, oblivious transfer is not possible without any computational assumptions in either the standard classical or quantum models \cite{dfss05}. Nevertheless, in \cite{dfrss07} it was shown that $1\textsf{-}2$ oblivious transfer is possible in the \textsf{BQSM}.

Recall, in the \textsf{BQSM}, an adversary $\mathcal{A}_{\texttt{s}}$ has access to unlimited resources at all times except at certain points. At these points, we say \emph{the memory bound applies}, and the adversary is forced to reduce its stored state to $\texttt{s}$-qubits.

The protocol presented by \cite{dfrss07} is secure against any adversary which can store a quarter of the qubits transmitted in the input state whereas an honest user requires no qmemory. The construction is as follows:

\begin{construct}[\otqot{m,\ell}$(s_0,s_1,c)$ \cite{dfrss07}]
\label{con:QOT}
The construction is as follows: 

{\small 
\begin{enumerate}
\item $\mathsf{S}$ chooses $x\in_R\{0,1\}^{m}$, $\theta \in_R\{+,\times \}^m$ and sends $|x_1\rangle_{\theta_1},|x_2\rangle_{\theta_2},...,|x_m\rangle_{\theta_m}$ to $\mathsf{R}$.
\item $\mathsf{R}$ measures all qubits in basis $[+,\times]_c$. Let $x'$ be the result. 
\item[*] \texttt{$\mathsf{R}$'s memory bound $\texttt{s}$ applies.}
\item $\mathsf{S}$ picks two hashing functions $f_0,f_1\in_R \textsf{H}$, announces $\theta$ and $f_0,f_1$ to $\mathsf{R}$ and outputs $e_0\coloneqq f_0(x|_{I_0})\oplus s_0$ and $e_1\coloneqq f_1(x|_{I_1})\oplus s_1$ where $I_b\coloneqq \{i:\theta_i=[+,\times]_b\}.$
\item $\mathsf{R}$ outputs $y\coloneqq e_c\oplus f_c(x'|_{I_c}).$
\end{enumerate} }
\end{construct}

\begin{theorem}[Propositions 4.5 and 4.6 in \cite{dfrss07}]
\label{QOT}
$\otqot{m, \ell}$ is perfectly receiver-secure and $\epsilon$-sender-secure against any dishonest receiver with qmemory bound of $\texttt{s}$-qubits for negligible (in $m$) $\epsilon$ if $\frac{m}{4}-2\ell -\texttt{s} \in \Omega (m)$.
\end{theorem}

\subsection{Algebra}
This section shows that it is difficult to learn a large matrix or polynomial using a small number of samples. Essentially, this allows us to construct $q$-time pads which will be highly useful in combination with our one-time compiler. Two constructions are provided: the first based on matrices and the other on polynomials.

The proof of the following lemma is straightforward. 
\begin{lemma}
\label{Raz 1}
Given samples $(a_1,b_1),...,(a_{m},b_{m})$ where $a_i \in \{0,1\}^n$, $b_i\in \{0,1\}$ and $m<n$. There are at least $2^{(n-m)}$ binary vectors ${v}$ such that ${a}_i\cdot {v}=b_i$ for all $i\in [m]$.
\end{lemma}

The following theorem generalizes Lemma \ref{Raz 1} to matrices. See App. \ref{app:Raz 2} for the proof.

\begin{theorem}
\label{Raz 2}
Let $M$ be an arbitrary $\ell \times n$ binary matrix. Let $A$ be an algorithm that is given as input: $({a}_1,b_1),...,({a}_{m},b_{m})$ and $(\hat{a}_1,\hat{b}_1),...,(\hat{a}_{p},\hat{b}_{p})$ where $a_i, \hat{a}_i \in \{0,1\}^n$, $b_i,\hat{b}_i\in \{0,1\}^{\ell}$, $m<n$, $p$ is a polynomial in $\ell$, $b_i=M\cdot {a}_i$ and $\hat{b}_i\neq M\cdot {\hat{a}}_i$. Then the following statements hold:
\begin{enumerate}
  \item For any vector ${a}\in \{0,1\}^n$ not in the span of $({a}_1,...,{a}_m)$, if $A$ outputs a guess $b'$ of $b\coloneqq M\cdot {a}$, then $\Pr{[b'=b]}= O(2^{-\ell})$.
  
  \item Let ${x}_0,{x}_1\in \{0,1\}^n$ be any two distinct vectors not in the span of $({a}_1,...,{a}_m)$. Choose $r\in_R \{0,1\}$. If $A$ is additionally given ${x}_0,{x}_1, {y}_r$ where ${y}_r\coloneqq M\cdot {x}_r$ and outputs a guess $r'$ then $\Pr{[r'=r]}\leq \frac{1}{2}+ O(2^{-\ell})$.
\end{enumerate}
\end{theorem}

The same idea applies to polynomials as well. The following is a result of the well-known Lagrange interpolation. See App. \ref{app:lagrange} for the proof. 

\begin{lemma}
\label{lagrange}
Let $\mathbb{F}$ be a field of order $2^{\ell}$. Let $(a_0,b_0),...,(a_{m},b_{m})$ be pairs such that $a_i,b_i\in \mathbb{F}$ and the $a_i$ are distinct. The number of polynomials $f$ of degree $d$ where $m<d$ that satisfy $f(a_i)=b_i$ for all $i\in [m+1]$ is $2^{(d-m)\ell}$.
\end{lemma}

By the same reasoning as in the case of matrices (Theorem \ref{Raz 2}), we get the following result, which is proven in App. \ref{app:lagrange 2}.

\begin{theorem}
\label{lagrange 2}
Let $\mathbb{F}$ be a field of order $2^{\ell}$ and let $f(x)\in \mathbb{F}[x]$ be a polynomial of degree $d$. Let $(a_0,b_0),...,(a_m,b_m)$ and $(\hat{a}_1,\hat{b}_1),...,(\hat{a}_{p},\hat{b}_{p})$ be samples such that $m<d$, $p$ is a polynomial in ${\ell}$, $b_i=f(a_i)$ and $\hat{b}_i\neq f(\hat{a}_i)$.
Then, the following holds:
\begin{enumerate}
  \item For any element $a\in \mathbb{F}\backslash \{a_0,...,a_m\}$, if $A$ outputs a guess $b'$ for $b\coloneqq f(a)$ then $\Pr{[b'=b]}= O(2^{-\ell})$.
  
  \item Let $x_0,x_1\in \mathbb{F}\backslash \{a_0,...,a_m\}$ be any two distinct elements. Choose $r\in_R \{0,1\}$. If $A$ is additionally given $x_0,x_1, y_r$ where $y_r\coloneqq f(x_r)$ and outputs a guess $r'$ then $\Pr{[r'=r]}\leq \frac{1}{2}+ O(2^{-\ell})$.
\end{enumerate}
\end{theorem}

\section{Definitions: Obfuscation and Variants in the {BQSM}}

In this section, we adapt the notions of obfuscation and one-time programs to the \textsf{BQSM} and introduce a related notion termed program broadcast. 

First of all, obfuscation is a form of program encryption with many variants, as introduced in \cite{ABD21}. These definitions are unsatisfactory for our setting since adversaries in the \textsf{BQSM} differ from those considered in black/grey-box obfuscation. Hence, we introduce the notion of \emph{\textsf{BQS} obfuscation}. 

\begin{definition}[\textsf{BQS} Obfuscation]
\label{BQSBB defn}
A algorithm $O$ is a \emph{$(r,\texttt{s})$-\textsf{BQS} obfuscator} of the class of classical circuits $\mathcal{F}$ if it QPT and satisfies the following:
\begin{enumerate}
\item (functionality) For any circuit $C\in \mathcal{F}$, the circuit described by $O(C)$ can be used to compute $C$ on an input $x$ chosen by the evaluator. 
\item For any circuit $C\in \mathcal{F}$, the receiver {
can learn an evaluation of $C$ using $O(C)$ with at most $r$ qubits of qmemory
}.
\item (security) For any computationally unbounded adversary $\mathcal{A}_{\texttt{s}}$ there exists a computationally unbounded simulator $\mathcal{S}_{\texttt{s}}$ such that for any circuit $C\in \mathcal{F}$, 
\begin{align*}|\Pr{[\mathcal{A}_{\texttt{s}}(O(C))=1]}-\Pr{[\mathcal{S}_{\texttt{s}}^{C}(|0\rangle^{\otimes \lvert C\rvert})=1]}\rvert \leq \negl[\lvert C\rvert] .\end{align*}
\end{enumerate}
\end{definition}

One-time programs, introduced in \cite{BGS13}, are similar to obfuscation but can only be used to learn a single evaluation of the program. We adapt this notion to the \textsf{BQSM}. 

\begin{definition}[\textsf{BQS} One-Time Program]
\label{def:BQS one-time}
An algorithm $O$ is a \emph{$(r,\texttt{s})$-\textsf{BQS} one-time compiler} for the class of classical circuits $\mathcal{F}$ if it is QPT, satisfies the first three conditions of Definition \ref{BQSBB defn}, and the following:
\begin{enumerate}
\setcounter{enumi}{3}
\item (security) For any computationally unbounded adversary $\mathcal{A}_{\texttt{s}}$ there exists a computationally unbounded simulator $\mathcal{S}_{\texttt{s}}$ such that for any circuit $C\in \mathcal{F}$
\begin{align*}|\Pr{[\mathcal{A}_{\texttt{s}}(O(C))=1]}-\Pr{[\mathcal{S}_{\texttt{s}}^{1C}(|0\rangle^{\otimes \lvert C\rvert})=1]}\rvert \leq \negl[\lvert C\rvert] .\end{align*}
\end{enumerate}
\end{definition}

{
Here $\mathcal{S}_{\texttt{s}}^{1C}$ denotes a simulator with access to at most one classical oracle query to $C$, following the notation $A^{qP}$ introduced in Sec.~\ref{sec:Preliminaries}. 
}

We introduce the notion of \textsf{BQS} program broadcast which is similar to one-time programs but additionally requires that multiple copies of the encrypted program can only be used to learn a limited number of evaluations. While one-time programs allow for symmetric key cryptography, program broadcast allows for powerful applications such as asymmetric cryptography and tokens. 

\begin{definition}[\textsf{BQS} Program Broadcast]
\label{def BQS program broadcast}
A \emph{$( q,\texttt{s}, k)$-\textsf{BQS} program broadcast} for the class of circuits $\mathcal{C}$ consists of the following QPT algorithms:
\begin{enumerate}
\item $\textsf{KeyGen}(1^\lambda,t_{\textnormal{end}}):$ Outputs a classical key $ek$. 
\item $\textsf{Br}(\s,ek,C):$ Outputs a quantum transmission $O_C$ for the circuit $C\in \mathcal{C}$ during broadcast time (before $ t_{\textnormal{end}}$). Outputs $ek$ after broadcast time.
\item $\textsf{Eval}(\langle O_C, ek\rangle, x):$ Outputs an evaluation $y$ on input $x$ from the transmission $\langle O_C, ek\rangle$ using $q$ qmemory.
\end{enumerate}
Correctness requires that for any circuit $C\in \mathcal{C}$ and input $x$,
\begin{align*} \Pr{\left[
\begin{tabular}{c|c}
 \multirow{2}{*}{$\textsf{Eval}(\langle O_C, ek\rangle, x)=C(x)\ $} &  $ek\ \leftarrow \textsf{KeyGen}(1^\lambda,t_{\textnormal{end}})$ \\ 
 & $O_C\ \leftarrow \textsf{Br}(\s,ek,C)$\\
 \end{tabular}\right]} \geq 1-\negl[\lambda] .
\end{align*}
Security requires that for any computationally unbounded adversary $\mathcal{A}_{\texttt{s}}$ there exists a computationally unbounded simulator $\mathcal{S}_{\texttt{s}}$ such that for any circuit $C\in \mathcal{C}$, and $ek\leftarrow \textsf{KeyGen}(1^\lambda,t_{\textnormal{end}})$,
\begin{align*}|Pr[\mathcal{A}_{\texttt{s}}^{\textsf{Br}(\s,ek,C)}(|0\rangle)=1]-Pr[\mathcal{S}_{\texttt{s}}^{kC}(|0\rangle^{\otimes \lvert C\rvert})=1]\rvert \leq \negl[\lambda] .\end{align*}
\end{definition}

Notice that security essentially requires that an adversary with access to polynomial outputs of the broadcaster can be simulated with access to $k$ oracle queries to $C$. The evaluation key $ek$ allows users to evaluate an output $O_C$ from the broadcast and only becomes public at the end of the broadcast. We introduce this key to ensure that an adversary cannot simply query the broadcast continuously to learn an unbounded number of evaluations. This key is revealed only at the end of the broadcast to ensure that the adversary needs to store the states $O_C$ in order to evaluate them with the key which limits the number of evaluations the adversary can learn. On the downside, this means the broadcast lasts for a limited time which is clearly unavoidable considering that adversaries are computationally unbounded. 

An interesting property of all these definitions is that it is sufficient to give the simulator access to only \emph{classical} oracle queries to the encrypted program. Indeed, we never need to deal with providing quantum oracle access in this paper. This attractive aspect to the \textsf{BQSM} is due to the natural ``classicalization'' of quantum states as a result of the qmemory bound.

\section{Definitions: Disappearing Cryptography}
\label{sec:dis and uncl}
In this section, we initiate the study of disappearing cryptography in the \textsf{BQSM}. These concepts were defined earlier in \cite{GZ21,g02} but we adapt the definitions to the \textsf{BQSM}. We define these notions for oblivious transfer, one-time programs, and asymmetric key encryption. 

Informally, a state $|\phi\rangle$ is \emph{disappearing} if any adversary that receives $|\phi\rangle$ cannot produce a state with the same `functionality' as $|\phi\rangle$ after a certain point in the transmission. In the \textsf{BQSM}, this point is when the adversary's qmemory bound applies. 

\subsection{Oblivious Transfer}
\label{sec:OT dis and unc}

We define disappearing security for oblivious transfer and show that Construction \ref{con:QOT} in the \textsf{BQSM} (Construction \ref{con:QOT} in Sec.~\ref{sec:OT}) satisfies this notion. This will be useful in showing later that our one-time program, encryption, and authentication schemes also satisfy disappearing security.

With regards to oblivious transfer, functionality refers to the first condition of Definition \ref{QOT}. In particular, we present the following experiment to test the disappearing security of a $1\textsf{-}2\ OT$ protocol.

\smallskip \noindent\fbox{%
  \parbox{\textwidth}{%
\textbf{Experiment} $\textsf{OT}^{\textsf{Dis}}_{\Pi,\mathcal{A}}({m},\ell)$:
Let $\Pi(m,S_0,S_1)$ be the oblivious transfer protocol with security parameter $m$ applied to strings $S_0,S_1\in \{0,1\}^\ell$.
\begin{enumerate}
\item Sample two strings $s_0,s_1\in_R\{0,1\}^\ell$ and a bit $b\in_R\{0,1\}$.
\item Protocol $\Pi(m,s_0,s_1)$ is executed between the experiment and adversary $\adv$. (Experiment is the sender and $\adv$ is the receiver of the OT).
\item After the execution ends, send $b$ to $\adv$.
\item $\adv$ outputs a guess $\hat{s}_b$.
\item The output of the experiment is 1 if $\hat{s}_b=s_b$ and 0 otherwise.
\end{enumerate}}}
\smallskip
\begin{definition}
A $1\textsf{-}2\ OT$ protocol $\Pi_{\texttt{s}}$ is \emph{disappearing} in the \textsf{BQSM} if for any adversary $\advs$,
\begin{align*}
  \Pr{[ \textsf{\em{OT}}^{\textsf{Dis}}_{\Pi_{\texttt{s}},\advs}({m},\ell)=1]}\leq \frac{1}{2}+\negl[\min(m,\ell)] .
\end{align*}
\end{definition}

\begin{theorem}
\label{disappearing OT}
Construction \ref{con:QOT} $(\otqot{m,\ell})$ is disappearing against adversaries with qmemory bound $\texttt{s}$ satisfying $m/4-2\ell -\texttt{s}\in \Omega(m)$.
\end{theorem}

\begin{proof}
This is a direct result of the proof of Theorem \ref{QOT}. It is shown that for any adversary $\adv_{\texttt{s}}$, after the memory bound applies, there exists $c\in\{0,1\}$ such that $s_{1-c}$ is close to random given the adversary's quantum state of $\texttt{s}$ qubits if $m/4-2\ell -\texttt{s}$ is positive and linear in $m$. Hence, except with negligible probability, the output of the experiment will be 0 if $b\neq c$. 
\qed     \end{proof}

\subsection{One-Time Programs}
\label{def dis otp}


To define the security properties of an disappearing one-time program $\mathcal{O}$, we present two experimental frameworks, similar to the experiments used to define disappearing grey-box obfuscation in \cite{GZ21}.

\begin{enumerate}
    \item $\mathsf{ExpAdv}_{\mathcal{A}, ch, \mathcal{O}}(C \in \mathcal{C}_\lambda, k)$:
    \begin{itemize}
        \item The experiment is comprised of an unspecified number of rounds. In any given round, one of the two following events occurs:
        \begin{itemize}
            \item \textbf{Interaction Round:} The adversary $\mathcal{A}$ engages in arbitrary communication with the challenger $ch$.
            \item \textbf{Stream Round:} The adversary $\mathcal{A}$ is provided with a fresh stream $\mathcal{O}(C)$ of the obfuscated circuit. Simultaneously, the challenger $ch$ is issued a unique tag signaling that a streaming event has taken place.
        \end{itemize}
        \item The challenger $ch$ maintains the authority to conclude the experiment at any moment by producing a bit $b \in \{0, 1\}$, which serves as the final output.
        \item Should the total number of stream rounds surpass the threshold $k$, the challenger $ch$ immediately outputs $0$ and terminates the session.
    \end{itemize}

    \item $\mathsf{ExpSim}_{\mathcal{S}, ch, \mathcal{O}}(C \in \mathcal{C}_\lambda, k, q)$:
    \begin{itemize}
        \item The experiment is comprised of an unspecified number of rounds:
        \begin{itemize}
            \item \textbf{Interaction Round:} The simulator $\mathcal{S}$ engages in arbitrary communication with the challenger $ch$.
            \item \textbf{Stream Round:} The simulator $\mathcal{S}$ is permitted to submit up to $q$ adaptive oracle queries to the circuit $C$ and receives the corresponding outputs. The challenger $ch$ is issued a unique tag signaling that a streaming event has taken place.
        \end{itemize}
        \item The challenger $ch$ may choose to end the experiment at any time by outputting a bit $b \in \{0, 1\}$, which serves as the program's output.
        \item Whenever the count of stream rounds exceeds $k$, the challenger $ch$ immediately outputs $0$ and terminates the session.
    \end{itemize}
\end{enumerate}

The interaction round essentially gives the challenger with a mechanism to send auxiliary data regarding $C$, such as an accepting input. Crucially, this auxiliary data may be revealed \textit{after} the obfuscated stream has been observed, at which point the stream effectively ceases to exist, preventing the adversary from further querying the program.

\begin{definition}[Disappearing One-Time Program]
    A $(r, \s)$-\textsf{BQS} one-time compiler $\mathcal{O}$ for a class of circuits $\mathcal{F}$ satisfies \emph{disappearing security} in the \textsf{BQSM} if for every challenger $\mathsf{ch}$ and any computationally unbounded adversary $\mathcal{A}_\s$, there exists a simulator $\mathcal{S}_\s$ with unbounded computational power such that for every circuit $C \in \mathcal{F}$:
\[
\left| \Pr[\mathsf{ExpAdv}_{\mathcal{A}_\s, \mathsf{ch}, \mathcal{O}}(C, k) = 1] - \Pr[\mathsf{ExpSim}_{\mathcal{S}_\s, \mathsf{ch}, \mathcal{O}}(C, k, q) = 1] \right| \leq \mathsf{negl}(\lambda),
\]
where $q = \mathsf{poly}(\lambda)$.
\end{definition}

\subsection{Encryption}
\label{sec:enc def}

We first recall the definition of a quantum asymmetric (public) key encryption scheme on classical messages. Note that the public key in this setting is quantum so multiple copies must be created and distributed due to the no-cloning theorem. Hence, we add an algorithm $\textsf{KeySend}$ that outputs a copy of the quantum public key when queried. In our security experiment, the adversary is allowed to receive a polynomial number of public key copies.

We also introduce an algorithm $\textsf{KeyReceive}$ which describes how to extract a reusable classical key from a public key copy to use for encryption. It is not generally required that users be able to extract a classical encryption key as the quantum public key may be used directly for encryption. However, constructions with reusable classical encryption keys are generally more practical. In this work, it is assumed that the quantum public keys can be distributed securely. Note that further work was done concurrently with this paper showing how to certify quantum keys from classical public-key infrastructure \cite{BS232}. 

\begin{definition}[Quantum Asymmetric Key Encryption]
A \emph{quantum asymmetric key encryption scheme} $\Pi$ over classical message space $\hil{M}$ consists of the following QPT algorithms: 
\begin{itemize}
  \item $\textsf{Gen}(1^\lambda):$ Outputs a private key $sk$. 
  \item $\textsf{KeySend}(sk):$ Outputs a quantum public key copy $\rho_{pk}$.
  \item $\textsf{KeyReceive}(\rho_{pk}):$ Extracts a key $k$ from $\rho_{pk}$. 
  \item $\textsf{Enc}(k,\mu):$ Outputs a ciphertext $\rho_{ct}$ for $\mu \in \hil{M}$.
  \item $\textsf{Dec}(sk, \rho_{ct})$: Outputs a message $\mu'$ by decrypting $\rho_{ct}$. 
\end{itemize}
\end{definition}

\begin{definition}[Correctness]
A quantum asymmetric key encryption scheme $\Pi$ is \emph{correct} if for any message $\mu \in \hil{M}$,
\begin{align*} \Pr{\left[
\begin{tabular}{c|c}
 \multirow{4}{*}{$\textsf{Dec}(sk,\rho_{ct})=\mu\ $} &  $sk\ \leftarrow \textsf{Gen}(1^\lambda)$ \\ 
 & $\rho_{pk}\ \leftarrow \textsf{KeySend}(sk)$\\
 & $k\ \leftarrow \textsf{KeyReceive}(\rho_{pk})$\\
 & $\rho_{ct}\ \leftarrow \textsf{Enc}(k,\mu)$\\
 \end{tabular}\right]} \geq 1-\negl[\lambda] .
\end{align*}
\end{definition}

We now present an experiment to test disappearing security under lunchtime chosen ciphertext attack (qDCCA1) in the \textsf{BQSM}. Recall, IND-CCA1 \cite{CS98} tests security against an adversary that can query the encryption oracle at any point in the experiment and can query the decryption oracles at any point before the challenge ciphertext is received. To satisfy \textit{disappearing} security under CCA1, we require that an adversary cannot win the IND-CCA1 experiment even if the private key is revealed after the challenge ciphertext is given. This does not make sense in the standard model since the adversary can simply store the challenge ciphertext and then decrypt it using the private key. However, in the \textsf{BQSM} this strong form of encryption is possible since the adversary may not be able to store the ciphertext for later use. In the experiment, we let $A^{\textsf{KeySend}(sk)}$ denote that the algorithm $A$ has access to polynomial number of outputs of $\textsf{KeySend}(sk)$.

\smallskip \noindent\fbox{%
  \parbox{\textwidth}{%
\textbf{Experiment} $\textsf{AsyK}^{\textsf{qDCCA1}}_{\Pi,\mathcal{A}}({\lambda})$:
\begin{enumerate}
  \item Sample a private key $sk\leftarrow \textsf{Gen}(1^\lambda)$ and bit $b\in_R \{0,1\}$.
  \item Generate a public key $\rho_{pk} \leftarrow \textsf{KeySend}(sk)$.
  \item Extract key $k \leftarrow \textsf{KeyReceive}(\rho_{pk})$.
  \item Adversary outputs two messages $m_0,m_1 \leftarrow \mathcal{A}^{\textsf{KeySend}(sk),\textsf{Enc}(k,\cdot),\textsf{Dec}(sk,\cdot)}$.
  \item Send $\adv^{\textsf{KeySend}(sk),\textsf{Enc}(k,\cdot)}$ the ciphertext $\rho_{ct} \leftarrow \textsf{Enc}(k,m_b)$.
  \item Give $\adv^{\textsf{KeySend}(sk),\textsf{Enc}(k,\cdot)}$ the private key $sk$.
  \item $\adv^{\textsf{KeySend}(sk),\textsf{Enc}(k,\cdot)}$ outputs a guess $b'$.
  \item The output of the experiment is $1$ if $b'=b$ and 0 otherwise. 
  \end{enumerate}}}
  \smallskip 

A stronger notion is indistinguishability under adaptive chosen ciphertext attacks (IND-CCA2) which allows the adversary in the security experiment unrestricted access to the encryption and decryption oracles except on the challenge ciphertext. This notion was later adapted to quantum ciphertexts in \cite{AGM18}. However, we could not achieve IND-CCA2 in the \textsf{BQSM}, as we found it difficult to prevent the adversary from performing small modifications to the challenge ciphertext without affecting the decryption outcome, and then submitting the modified ciphertext to the decryption oracle. Note that such an attack does not require quantum memory since the adversary has access to the decryption oracle as it receives the challenge ciphertext in this scenario. Another variant is indistinguishability under chosen quantum ciphertext attacks \cite{BZ13,GSM20}, which allows the adversary to query the oracles on a superposition. However, we did not explore this notion and, in this work, we focus on CCA1-security and leave it as an open question to achieve these stronger notions. That being said, qDCCA1-security comes very close to CCA2-security since the adversary has the ability to decrypt when the private key is revealed. 

We can similarly construct disappearing experiment in the symmetric key case, denoted $\textsf{SymK}^{\textsf{qDCCA1}}_{\Pi,\mathcal{A}}({\lambda})$ by deleting steps 2 \& 3 and replacing $k$ with $sk$ in the asymmetric experiment.

\begin{definition}[Security]
An asymmetric key encryption scheme $\Pi_{\texttt{s}}$ satisfies \emph{disappearing indistinguishability against chosen ciphertext attacks (qDCCA1)} if for any adversary $\advs$,
\begin{align*}
  \Pr{[\textsf{AsyK}^{qDCCA1}_{\Pi_{\texttt{s}},{\advs}}({\lambda})=1]} &\leq \frac{1}{2}+\negl[\lambda] ,
\end{align*}
\end{definition}

\subsection{Authentication}
\label{sec:sig def}

We define a signature scheme with quantum keys on classical messages. 

\begin{definition}[Signature]
A \emph{signature scheme} with quantum public keys over a classical message space $\hil{M}$ consists of the following tuple of QPT algorithms:
\begin{itemize}
  \item $\textsf{Gen}(1^\lambda):$ Outputs a private key $sk$.
  \item $\textsf{KeySend}(sk):$ Output a quantum public key $\rho_{pk}$.
  \item $\textsf{KeyReceive}(pk):$ Extract key $k$ from $\rho_{pk}$. 
  \item $\textsf{Sign}( sk,\mu):$ Outputs a signature $\rho_{\sigma}$ for message $\mu\in \hil{M}$.
  \item $\textsf{Verify}(k, \mu,\rho_{\sigma}):$ Verifies $\rho_{\sigma}$ is a signature for $\mu$ and correspondingly outputs a bit $b$. 
\end{itemize}
\end{definition}

A \emph{message authentication scheme (\textsf{MAC})} is the same except without a public key so verification is also performed using the private key. 

\begin{definition}[Correctness]
We say a signature scheme with quantum public keys on classical messages is \emph{correct} if for any message $\mu\in \hil{M}$,
\begin{align*} \Pr{
\left[\begin{tabular}{c|c}
 \multirow{4}{*}{$\textsf{Verify}(k,\mu, \rho_{\sigma})=1\ $} & $sk\ \leftarrow \textsf{Gen}(1^\lambda)$ \\
  & $\rho_{pk}\ \leftarrow \textsf{KeySend}(sk)$\\
 & $k\ \leftarrow \textsf{KeyReceive}(\rho_{pk})$\\
 & $\rho_{\sigma}\leftarrow \textsf{Sign}(sk,\mu)$
 \end{tabular}\right]} \geq 1-\negl[\lambda] .
\end{align*}
\end{definition}

Informally, a signature is disappearing if an adversary cannot produce a valid signature for a message $\mu$ even if it has received a signature for $\mu$ earlier.

We present the following experiment to test disappearing unforgeability under chosen message and verification attacks. 

\smallskip \noindent\fbox{%
  \parbox{\textwidth}{%
\textbf{Experiment} $\textsf{Sign}^{\textsf{qDCMVA}}_{\Pi, \mathcal{A}}({\lambda})$:
\begin{enumerate}
  \item Sample a private key $sk\leftarrow \textsf{Gen}(1^\lambda)$.
  \item Generate a public key $\rho_{pk} \leftarrow \textsf{KeySend}(sk)$.
  \item Extract key $k \leftarrow \textsf{KeyReceive}(\rho_{pk})$.
  \item Run $\mathcal{A}^{\textsf{KeySend}(sk),\textsf{Sign}(sk,\cdot),\textsf{Verify}(k,\cdot)}$.
  \item Adversary outputs a pair $(\mu, \rho_{\sigma} )\leftarrow \mathcal{A}^{\textsf{KeySend}(sk),\textsf{Verify}(k,\cdot)}$.
  \item The output of the experiment is $1$ if $\textsf{Verify}(k,\mu, \rho_{\sigma})=1$ and 0 otherwise.
\end{enumerate}}}
\smallskip

We can similarly construct disappearing security for \textsf{MAC} schemes, denoted $\textsf{MAC}^{\textsf{qDCMVA}}_{\Pi, \mathcal{A}}({\lambda})$, by deleting step 2 \& 3 and replacing $k$ with $sk$ in the signature experiment.

\begin{definition}[Security]
A signature scheme $\Pi_{\texttt{s}}$ satisfies \emph{disappearing security under chosen message and verification attacks (qDCMVA)} if for any adversary $\advs$, 
\begin{align*}
  \Pr{[\textsf{Sign}^{\textsf{qDCMVA}}_{\Pi_{\texttt{s}},\advs}({\lambda})=1]}&\leq \negl[\lambda] ,\end{align*}
\end{definition}

\section{One-Time Programs}
\label{sec:one-time short}

In this section, we give the result for an information-theoretic secure one-time program in the \textsf{BQSM}. The construction is provided in App.~\ref{sec:one-time construction} as it is long and technical. Roughly speaking, it follows the same ideas as Kilian's construction for oblivious circuit evaluation \cite{k88} based on $1\textsf{-}2$ oblivious transfer. For our setting, we do not need to assume oblivious transfer as Construction \ref{con:QOT} already achieves this unconditionally. However, the proof in our setting is more difficult since constructing a simulator requires determining the chosen string received in each oblivious transfer transmission employed in the one-time program. We provide a new extraction result for this purpose:

\begin{lemma}[Extraction]
\label{determin C}
Let $A$ be a quantum circuit with $n$-qubit input and let $x\in_R \{0,1\}^n$ and $\theta \in_R \{+,\times\}^n$ be chosen uniformly at random. Run $A$ on $|x\rangle_\theta$, and let $X$ represent the random variable for $x$ from the perspective of $A$ at the end of the computation. Then, there exists a computationally-unbounded algorithm $S$ that can (inefficiently) approximate the distribution $X$, given $A$ and the measurement results of the computation $A(\ket{x}_\theta)$. 
\end{lemma}

{
The simulator is given a classical description of $A$ and the classical measurement results from one execution of $A(\ket{x}_\theta)$. It is not given multiple identical copies of the unknown input state. The proof uses classical descriptions of the possible states $\ket{x'}_{\theta'}$ to compute the probabilities of the observed outcomes. These computations do not require further copies of the state received by the adversary.
}

To prove this result, first we state Kraus' Theorem which classifies all operations that a quantum adversary can perform.  

\begin{theorem}[Kraus' Theorem \cite{NC02}]
\label{Kraus}
Let $\Phi$ be a quantum operation between two Hilbert spaces $\hil{H}$ and $\hil{G}$ of dimension $n$ and $m$, respectively. Then there exists matrices $\{B_i\}_{1\leq i\leq nm}$ which map $\hil{H}$ to $\hil{G}$ such that for any state $\rho$
\begin{align*}
  \Phi(\rho)=\sum_{1\leq i\leq nm}B_i\rho B_i^{*} .
\end{align*}
\end{theorem}

Now we are ready to prove Lemma \ref{determin C}

\begin{proof}
The first step is to convert (or approximate) the quantum circuit $A$ with a classical version $A^c$. A $n$-qubit state can be represented as a vector over the complex numbers while a quantum operation from $\hil{H}_{N}\rightarrow \hil{G}_{N'}$ can be represented with $NN'$ matrices of dimension $N\times N'$ by Kraus' Theorem \ref{Kraus}. Hence, $S$ can construct $A^c$ and evaluate it on any input. 

Assume the observed measurement results during the evaluation of $A$ on $|x\rangle_\theta$ are $m_1,m_2,...,m_k$. $S$ runs $A^c$ on every of the input $\ket{x'}_{\theta'}$ in vector form for $x'\in \{0,1\}^n$ and $\theta'\in \{+,\times\}^n$ and computes the probability of obtaining the measurement results $m_1,..., m_k$, i.e., $\Pr{[m_1,...,m_k|X=x'\Theta=\theta']}$ as described in the following. 

$S$ computes the evolution of the vector representation of $\ket{x'}_{\theta'}$ during the evaluation on $A_c$ by performing appropriate matrix multiplications (Kraus' Theorem \ref{Kraus}) and for the $i^{th}$ measurement, $S$ computes the probability,
\begin{align*}
  p(m_i)= v^{*}M_{m_i}^{*}M_{m_i}v
\end{align*}
where $v$ is the vector representing the state prior to the measurement and $M_i$ is the POVM element associated with the measurement outcome $m_i$. $S$ then computes the post-measurement state $\frac{M_{m_i}v}{p(m_i)}$ and continues the computation. At the end $S$ obtains an approximation of $\Pr{[m_1,...,m_k|X=x'\Theta=\theta']}$. $S$ repeats this computation for every possible value of $\theta'$ to obtain $\Pr{[m_1,...,m_k|X=x']}$ after summing over all values of $\theta'$. Next, by Baye's Theorem, 
\begin{align*}
\Pr{[X=x'|m_1,...,m_k]} &= \frac{\Pr{[X=x']}}{\Pr{[m_1,...,m_k]}}\Pr{[m_1,...,m_k|X=x']} \\
&=\frac{\Pr{[m_1,...,m_k|X=x']}}{2^n\Pr{[m_1,...,m_k]}} .
\end{align*}
so normalization gives the required distribution.

\qed     \end{proof}

Applied to an oblivious transfer transmission $1\textsf{-}2\textsf{-}\textsf{QOT}(s_0,s_1)$, Lemma \ref{determin C} implies that a simulator can learn the distribution of the two encoded strings $s_0$ and $s_1$ from the algorithm of the adversary. Based on the entropy of the calculated distributions, the simulator can determine which 1-out-of-2 strings the adversary is oblivious to. The rest of the construction basically just adapts Kilian's ideas to the quantum realm; the details of which are not the main focus of this paper and thus are given in App.~\ref{sec:one-time construction}.

Consequently, we construct a one-time program compiler, which we denote as $\mathcal{O}_{\texttt{s}}$, in the \textsf{BQSM}. Note that, by definition, $\mathcal{O}_{\texttt{s}}$ also acts as a \textsf{BQS} obfuscator. Our construction is also disappearing meaning a one-time program cannot be evaluated after the transmission ends. This is a direct result of disappearing properties of Construction \ref{con:QOT} for $1\textsf{-}2$ oblivious transfer which we showed in Sec.~\ref{sec:OT dis and unc}.

\begin{theorem}
\label{main poly}
There exists an algorithm $\mathcal{O}_{\texttt{s}}$ that is a disappearing information-theoretically secure $(0,\texttt{s})$ one-time compiler for the class of polynomial classical circuits against any computationally unbounded adversary with $\texttt{s}$ qmemory bound. 
\end{theorem}

\section{Impossibility of Non-Disappearing Obfuscation}
\label{sec:impossible}

A point of discussion is the disappearing nature of our protocol. After the transmission ends, the program would disappear due to the user's limitations, and re-evaluating the program on a new input would not be possible. On one hand, this extra security allows for further applications like disappearing ciphertexts (see Sec.~\ref{sec:encryption}). On the other hand, it is inconvenient when users wish to continually perform evaluations. In this scenario, the sender is required to send the one-time program multiple times whenever requested by the user. It is natural to wonder whether it is possible to construct one-time programs or obfuscations which do not disappear. In the rest of this section, we show that this is generally not possible. There are 3 different scenarios that need to be considered separately: 
\begin{enumerate}
  \item (Black-box obfuscation in (computational) \textsf{BQSM}) $\mathcal{A}_{\texttt{s}}$ and $\mathcal{S}_{\texttt{s}}$ are QPT algorithms with $\texttt{s}$ qmemory bound.
  \item (Grey-box obfuscation in (computational) \textsf{BQSM}) $\mathcal{A}_{\texttt{s}}$ is QPT and $\mathcal{S}_{\texttt{s}}$ is computationally unbounded, and both have $\texttt{s}$ qmemory bound.
  \item (\textsf{BQS} obfuscation) $\mathcal{A}_{\texttt{s}}$ and $\mathcal{S}_{\texttt{s}}$ are both computationally unbounded with $\texttt{s}$ qmemory bound. 
\end{enumerate}

Note that there exists a plausible construction in the computational setting for grey-box obfuscation \cite{BCK17}. This construction can be adapted to the \textsf{BQSM} setting, since it is classical, to obtain non-disappearing grey-box obfuscation. Given this result, it seems likely that the second case is not impossible. The rest of this section shows that non-disappearing obfuscation is impossible in the first and third cases. 

If an obfuscation is non-disappearing then it means that an honest user can extract some state from the obfuscation that can be stored and used to evaluate the program. In particular, it is sufficient to show that there does not exist a state $\rho_C$ that: $(1)$ can be extracted from an obfuscation of $C$, $(2)$ can be used to evaluate $C$ and $(3)$ can be feasibly stored. The last condition essentially requires that the size of $\rho_C$ is smaller than $\s$ qubits, where $\s$ is the qmemory required to threaten the security of the scheme. 

\begin{theorem}
\label{bb}
Let $O_{\texttt{s}}$ be a quantum black-box obfuscator of the class $\mathcal{C}$ of polynomial classical circuits. Assuming LWE is quantum hard, there do not exist QPT algorithms $\mathcal{A}_{\texttt{s}},\advs'$ such that for any $C\in \mathcal{C}$:
\begin{enumerate}
  \item $\mathcal{A}_{\texttt{s}}$ outputs $\rho_C\leftarrow \mathcal{A}_{\texttt{s}}(O_{\texttt{s}}(C))$ where $\rho_C$ is a cq-state with at most $\texttt{s}$ qubits.
  \item For any input $x$, 
\begin{align*}
\| \mathcal{A}'_{\texttt{s}}(\rho_C \otimes \proj{x})-\proj{C(x)} \|_1 \leq \negl[\lvert C\rvert] .
\end{align*}
\end{enumerate}
\end{theorem}

\begin{proof}
    Assume that there exist algorithms $\mathcal{A}_{\texttt{s}}$ and $\advs'$ satisfying the above conditions. Notice the $\texttt{s}$ qubit state $\rho_C$ acts as a black-box obfuscation of $C$. Importantly, $\rho_C$ does not disappear when the qmemory bound of $\mathcal{A}'_{\texttt{s}}$ applies since $\rho_C$ is $\texttt{s}$ qubits. Under the assumption that LWE is quantum hard, Alagic et al. \cite{ABD21} showed that quantum black-box obfuscation for classical circuits is impossible against QPT adversaries that can store the entire obfuscation. 

\qed     \end{proof}

For the third case, we show that non-disappearing \textsf{BQS} obfuscation for the small class of interval functions is impossible unconditionally. An interval function is a function that evaluates to 1 on some interval, such as $x\in [n_0,n_1]$ where $n_0<n_1$, and 0 everywhere else. This is a very strong impossibility result since interval functions are a very simple class. 

\begin{theorem}
\label{qbb}
Let $O_{\texttt{s}}$ be a \textsf{BQS} obfuscator of the class $\mathcal{I}$ of interval functions. There do not exist (computationally unbounded) algorithms $\mathcal{A}_{\texttt{s}},\advs'$ such that for any $C\in \mathcal{I}$: 
\begin{enumerate}
  \item $\mathcal{A}_{\texttt{s}}$ outputs $\rho_C\leftarrow \mathcal{A}_{\texttt{s}}(O_{\texttt{s}}(C))$ where $\rho_C$ is a cq-state with at most $\texttt{s}$ qubits.
  \item For any input $x$, 
\begin{align*}
\| \mathcal{A}'_{\texttt{s}}(\rho_C \otimes \proj{x})-\proj{C(x)} \|_1 \leq \negl[\lvert C\rvert] .
\end{align*}
\end{enumerate}
\end{theorem}

\begin{proof}
Assume there exist algorithms $\adv_{\texttt{s}}$ and $\advs'$ satisfying the above conditions. Let $I_{[x_0,y_1]}(x)\in \mathcal{I}$ denote an interval function which takes input $x$ and strings $y_0,y_1$ of length $n$ and outputs 1 if $x\in [y_0,y_1]$ and 0 otherwise. There exists a general circuit for this function which can be used for all values $y_0$ and $y_1$. 

As a result, there exists a super-polynomial function $G$ such that for any interval function $I_{[y_0,y_1]}$, $\mathcal{A}_{\texttt{s}}$ outputs $\rho_{I_{[y_0,y_1]}}\leftarrow \mathcal{A}_{\texttt{s}}(O_{\texttt{s}}(I_{[y_0,y_1]}))$ and,
\begin{align*}
\| \mathcal{A}'_{\texttt{s}}(\rho_{I_{[y_0,y_1]}}\otimes |x\rangle\langle x|)-|I_{[y_0,y_1]}(x)\rangle \langle I_{[y_0,y_1]}(x)| \|_1 \leq \frac{1}{G(n)} .
\end{align*}

Pick a random integer $k\in [0,G(n)^{1/4}]$ and let $C\coloneqq I_{[\frac{k2^{n+2}}{G(n)}, \frac{(k+1)2^{n+2}}{G(n)}]}$. We now show that $\mathcal{A}_{\texttt{s}}$ and $\adv'$ can learn more about $C$ from $O_{\texttt{s}}(C)$ than a simulator can with polynomial number of queries to an oracle of $C$.

Given $O_{\texttt{s}}(C)$, the adversary produces a state $\rho_C\leftarrow \mathcal{A}_{\texttt{s}}(O_{\texttt{s}}(C))$. For any input $x_0$, $\mathcal{A}'_{\texttt{s}}$ can obtain a state near $|C(x_0)\rangle \langle C(x_0)|$ using $\rho_C$. Measuring the resulting state gives the evaluation $C(x_0)$ with high probability. Next, by the gentle measurement lemma \cite{W99}, $\advs'$ can perform the inverse operations and obtain a state $\rho_1$ such that,
\begin{align*}
  \| \rho_1-\rho_C\|_1 \leq \frac{1}{G(n)}
\end{align*}
since trace is invariant under unitary transformations. So, by the triangle inequality, for any input $x_1$,
\begin{align*}
\| \mathcal{A}'_{\texttt{s}}(\rho_1\otimes |x_1\rangle\langle x_1|)-|C(x_1)\rangle \langle C(x_1)| \|_1 \leq \frac{2}{G(n)} .\end{align*}
So $\advs'$ can obtain $C(x_1)$ with high probability. 

Therefore, this process can be repeated and $\advs'$ can, with probability at least $(1-\frac{m^2}{G(n)})^{m}$ obtain $m$ correct evaluations of $C$. For $m=G(n)^{1/4}$, this probability is larger than $1/3$, implying that $\adv$ has a non-negligible chance of determining $k$. On the other hand, a simulator with polynomial queries to the program $C$ has a negligible chance of determining $k$. 

\qed     \end{proof}

\section{Program Broadcast}
\label{sec:broadcast}

In this section, we construct a \textsf{BQS} program broadcaster as introduced in Definition \ref{def BQS program broadcast}. This will allow us to tackle asymmetric key encryption and authentication in later sections. Let $\mathcal{C}_{m}=\{\mathcal{C}_{n,m}\}_{n\leq 2^m}$ where $\mathcal{C}_{n,m}$ is a set of polynomial-size circuits in $n$, of input size $n$, and output size $m$. Note that any polynomial-size circuits belongs to such a class by adding null outputs to ensure that $n\leq 2^m$.  

\begin{construct}[\textsf{BQS} Program Broadcaster]
\label{con:broadcaster}
{\small The $(12m,\texttt{s},\frac{\texttt{s}}{2m})$-\textsf{BQS} program broadcast scheme for the class $\mathcal{C}_m$ until time $t_{\textnormal{end}}$ is as follows:
\begin{itemize}
  \item $\textsf{KeyGen}(1^\lambda,t_{\textnormal{end}})$: Choose uniformly at random a $(12m)\times (12m)$ binary matrix $ M$ and let $F(x)\coloneqq M\cdot x$ be the corresponding program. Choose a two-universal hash function $H\in_R \textsf{H}_{12m,m}$. Output $ek\coloneqq (M,H,t_{\textnormal{end}})$.
 
  \item $\textsf{Br}(\s,ek,P):$
  Let $P\in \mathcal{C}_m$. If queried before time $t_{\textnormal{end}}$, then:
  \begin{enumerate}
    \item Randomly choose $\mathbf{r},\mathbf{x} \in_R \{0,1\}^{12m} $. 
    \item Compute $\theta\coloneqq F(\mathbf{r})$ and $c\coloneqq H(\mathbf{x})$.
    \item Send $|\mathbf{x}\rangle_{\mathbf{\theta}}\coloneqq |x_1\rangle_{\theta_1}...|x_{12m}\rangle_{\theta_{12m}}$.
    \item Send $\mathcal{O}_{\texttt{s}}(P\oplus c)$. 
    \item Send $\mathbf{r}$.
  \end{enumerate}
 The entire transmission is denoted as $O^{r,x}_{P}$.

 If queried after time $t_{\textnormal{end}}$, then apply the qmemory bound and output $ek$.
  
\item $\textsf{Eval}(\langle O^{r,x}_P,ek \rangle,v):$
\begin{enumerate}
    \item Store $|\mathbf{x}\rangle_{\mathbf{\theta}}$.
    \item Evaluate the one-time program $\mathcal{O}_{\texttt{s}}(P\oplus c)$ on $v$ to obtain $P(v)\oplus c$.
    \item After $\mathbf{r} $ and $F$ are received, compute $\theta$. 
    \item Measure $|\mathbf{x}\rangle_{\theta}$ in the basis $\theta$ to obtain $\mathbf{x}$. 
    \item Use $H$ to compute $c=H(\mathbf{x})$ and obtain $P(v)$. \end{enumerate}
\end{itemize}}
\end{construct}

\begin{theorem}[Security]
\label{broadcast security}
Construction \ref{con:broadcaster} is a $( 12m, \texttt{s},\frac{\texttt{s}}{2m})$-\textsf{BQS} program broadcaster for the class $\mathcal{C}_m$.
\end{theorem}

\begin{proof}
Let $P\in \mathcal{C}_m$. It is clear that the receiver requires $12m$ qmemory to learn a single evaluation of $P$ from the broadcast.

In terms of security, an adversary $\mathcal{A}_{\texttt{s}}$ can receive a polynomial number, say $p$, outputs of the broadcast $\textsf{Br}(\s, ek,P)$. From these states, the adversary obtains $(|\mathbf{x}_i\rangle_{\mathbf{\theta}_i})_{i\in [p]}$ and one-time programs $(\mathcal{O}_{\texttt{s}}(P\oplus c_i))_{i\in [p]}$ where $c_i\coloneqq H(\mathbf{x}_i)$. 

$\mathcal{A}_{\texttt{s}}$ can be simulated by $\mathcal{S}_{\texttt{s}}$ which has access to a single oracle query to $P\oplus c_i$ for each $i\in [p]$ by the security of one-time protocol (Theorem \ref{main poly}). The lemma below shows that $\mathcal{S}_{\texttt{s}}$ can learn at most $\frac{\texttt{s}}{2m}$ values in $\{c_i\}_{i\in [p]}$. This implies that $\adv_{\texttt{s}}$ can learn at most $\frac{\texttt{s}}{2m}$ evaluations of $P$ from the broadcast thus achieving program broadcast security. 

\begin{lemma}
\label{distinguish}
$\mathcal{S}_{\texttt{s}}$ can distinguish at most $\frac{\texttt{s}}{2m}$ terms in $\{c_i\}_{i\in [p]}$ from random. 
\end{lemma}
\begin{proof}
Assume that there exists $\frac{\texttt{s}}{2m}$ values $\mathbf{x}_{i_1},...,\mathbf{x}_{i_{\frac{\texttt{s}}{2m}}}$ that are distinguishable from random. 

Instead of the broadcaster sending $|\mathbf{x}_i\rangle_{\mathbf{\theta}_i}$, a standard purification argument shows it is sufficient to show security for the protocol with the following modification. For each qubit supposed to be sent, the broadcaster instead prepares an EPR state $\frac{1}{\sqrt{2}}(|00\rangle+|11\rangle)$ and sends one half to $\mathcal{S}_{\texttt{s}}$ and keeps one half. Then the broadcaster measures its halves of the EPR pairs in a random BB84 basis in $\{+,\times \}^{12m}$ after the memory bound of $\mathcal{S}_{\texttt{s}}$ applies. Let $\Theta_{i}$ and $X_{i}$ be the random variables representing the broadcaster's choice of measurement basis and outcome. Let $X_I\coloneqq X_{i_1}X_{i_2}...X_{i_{\frac{\texttt{s}}{2m}}}$ and $\Theta_I\coloneqq \Theta_{i_1}\Theta_{i_2}...\Theta_{i_{\frac{\texttt{s}}{2m}}}$. The Uncertainty Relation (Lemma \ref{uncertinity relation}) implies that for $\gamma =\frac{1}{12}$, there exists $\epsilon$ (negligible in $m$) such that $H_{\infty}^{\epsilon}(X_I|\Theta_I)\geq (\frac{1}{2}-2\gamma) (6\texttt{s})=2\texttt{s}$. Let $B$ be the random $\texttt{s}$ qubit state the adversary stores when the memory bound applies. By the chain rule for min-entropy (Lemma \ref{chain}), \begin{align*}
  H_{\infty}^{\epsilon}(X_I|\Theta_IB)\geq H_{\infty}^{\epsilon}(X_I|\Theta_I)-\texttt{s}\geq \texttt{s} .
\end{align*}
Hence, by the Generalized Min-entropy Splitting Lemma \ref{splitting cor}, there exists a random variable $C$ and negligible value $\epsilon'>0$ such that,
\begin{align*}
H^{\epsilon'}_{\infty}(X_{i_C}|\Theta_IC B)\geq 2m-4 .
\end{align*}
By the Privacy Amplification Theorem \ref{privacy amplification}, \begin{align*}
  \delta(\rho_{H(X_{i_C})CH \Theta_IB},\mathds{1}\otimes &\rho_{C H\Theta_I B})\\
  &\leq \frac{1}{2}2^{-\frac{1}{2}(H^{\epsilon}_{\infty}(X_{i_C}|\Theta_ICB)_{\rho}-m)}+\epsilon =O(2^{-m})=\negl .
\end{align*} 
The final equality is because $m\geq \lg^2(n)$ which means $2^{-m}$ is negligible in $n$. This gives a contradiction so there is less than $\frac{\texttt{s}}{2m} $ values $c_i$ which are distinguishable from random. 

\qed     \end{proof}
This completes the proof of Theorem \ref{broadcast security}.

\qed     \end{proof}

\section{Encryption Schemes}
\label{sec:encryption}

We now construct an asymmetric encryption scheme in the \textsf{BQSM}. We first build a symmetric scheme and then upgrade it to an asymmetric scheme with quantum public keys.

\subsection{Symmetric Key Encryption}
\label{sec:enc}

In this section, we present a private key encryption scheme that satisfies disappearing indistinguishability under chosen ciphertext attacks (qDCCA1). 

\begin{construct}
\label{cca}
{\small The symmetric key encryption scheme $\Pi_{\textsf{SymK}}$ against adversaries with qmemory bound $\texttt{s}$ is as follows. 
\begin{itemize}
  \item $\textsf{Gen}(1^\lambda)$: Let $m\coloneqq \lceil (\lg \lambda)^{3/2}\rceil$. Choose at random strings $q,w\in_R \{0,1\}^m$, a $(2m) \times (2m)$ invertible binary matrix $M$ and string ${z}\in_R \{0,1\}^{2m}$. Define $S({x})\coloneqq M\cdot {x}+{z}$, and $S'(x)\coloneqq qx+w \mod{2^m}$. The private key is $k\coloneqq (q,w,M,{z})$. 
  \item $\textsf{Enc}(\texttt{s}, k, \mu )$: Let $\ell\coloneqq \lvert \mu \rvert$. Choose strings $a,b\in_R \{0,1\}^m$ and define $f(x) \coloneqq ax+b \mod{2^m}$. Construct the following program:
  \begin{align*}
    E_{k,f,\mu}(y)=\begin{cases} 
                       S'(f(x)) \|\mu & \textnormal{if ${y} =S(x\|f(x))$}\\
                       \perp & \textnormal{otherwise}.\\
                    \end{cases}\\
   \end{align*}
The encryption is $\rho_{ct} \leftarrow \langle \mathcal{O}_{\texttt{s}}(f),\mathcal{O}_{\texttt{s}}(E_{k,f,\mu})\rangle$ sent in sequence; first $\mathcal{O}_{\texttt{s}}(f)$ and then $\mathcal{O}_{\texttt{s}}(E_{k,f,\mu})$. 
\item $\textsf{Dec}(k, \rho_{ct})$: Check $\rho_{ct}$ has the correct format (see note below).
Evaluate the first one-time program on a random input $v$ to obtain $f(v)$. Evaluate the second one-time program on $S(v\|f(v))$. If the output is of the form $S'(f(v))\|\hat{\mu}$ (for some string $\hat{\mu}$) then output $\hat{\mu}$ and $\perp$ otherwise.
\end{itemize}}
\end{construct}

\textit{Note: For decryption, the receiver is required to check that the ciphertext has the correct format. To be clear, this means that the received ciphertext is of the form $\langle \mathcal{O}({f}),\mathcal{O}({E})\rangle$ for some $\mathbf{NC}^1$ functions ${f}:\{0,1\}^m \rightarrow \{0,1\}^m$ and ${E}:\{0,1\}^{2m}\rightarrow \{0,1\}^{m +\ell}$. The receiver can check that the ciphertext is of this form by inspecting the length of the received data and by receiver security of one-time programs (Theorem \ref{receiver security} in App.~\ref{sec:one-time construction}). Any ciphertext not of this form is immediately rejected as invalid.}


\begin{theorem}[Security]
\label{construction 1}
Construction \ref{cca} ($\Pi_{\textsf{SymK}}$) satisfies qDCCA1 security against computationally unbounded adversaries with qmemory bound $\texttt{s}$.
\end{theorem}

\begin{proof}
An adversary $\mathcal{A}_{\texttt{s}}$ in the qDCCA1-security experiment requests a polynomial number, say $Q_e$, of encryption queries and a polynomial number, say $Q_d$, of decryption queries. Denote the ciphertexts produced by all encryption queries to the oracle as 
$\langle \mathcal{O}({f}_i),\mathcal{O}({E}_{k,f_i,\mu_i})\rangle_{i\in [Q_e]}$ and denote the decryption queries submitted to the oracle as $\langle \mathcal{O}(\hat{f}_j),\mathcal{O}(\hat{E}_j)\rangle_{j\in [Q_d]}$ for $\mathbf{NC}^1$ functions $\hat{f}_j:\{0,1\}^m \rightarrow \{0,1\}^m$ and $\hat{E}_j:\{0,1\}^{2m}\rightarrow \{0,1\}^{m+\ell}$. Any decryption query not of this form is immediately rejected as an invalid ciphertext. 

$\advs$ cannot learn anything from the encryption queries alone as such attacks can be simulated with only a single oracle query to each program in $({E}_{k,f_i,\mu_i})_{i\in [Q_e]}$. All such queries will yield $\perp$ except with negligible probability by the security of one-time programs (Theorem \ref{main poly}). 

Now, consider what $\advs$ learns from the first decryption query $\langle \mathcal{O}(\hat{f}_0),\mathcal{O}(\hat{E}_0)\rangle$. 
Suppose $\advs$ requested the encryption queries $\langle \mathcal{O}({f}_i),\mathcal{O}({E}_{k,f_i,\mu_i})\rangle_{i\in [d_0]}$ prior to requesting the first decryption query. These ciphertexts may be utilized in the construction of the decryption query. Let $\mathsf{A}_0$ be the sub-algorithm of $\advs$ which produces the first decryption query using these ciphertexts. 
Hence, we write $\langle \ot(\hat{f}_0), \ot(\hat{E}_0)\rangle \leftarrow \mathsf{A}_0(\langle \mathcal{O}({f}_i),\mathcal{O}({E}_{k,f_i,\mu_i})\rangle_{i\in [d_0]})$. 

As a side note, prior to sending $\ot(\hat{E}_0)$, the memory bound for $(\ot(f_i))_{i\in [d_0]}$ must have already been applied; thus there exists values $(v_i)_{i\in [d_0]}$ such that $\advs$ can only distinguish $f_i(x)$ from random on $x=v_i$ for all $i\in [d_0]$. This is because these functions are of the form $f_i(x)=a_ix+b_i \ \mod 2^m$ so learning one evaluation is not sufficient to distinguish the evaluation on another input from random. 

When the oracle receives the first decryption query, it evaluates the first one-time program on a random input say $\hat{v}_0$ to get $\hat{f}_0(\hat{v}_0)$. Next, it evaluates the second one-time program on $S(\hat{v}_0\|\hat{f}_0(\hat{v}_0))$ and gets an output say $\hat{y}_0\|\hat{\mu}_0$ (or $\perp$). Let $\mathsf{Or}_0$ denote the sub-algorithm of the oracle which performs this evaluation. Then this entire interaction can be described as $\hat{y}_0\|\hat{\mu}_0\leftarrow \mathsf{Or}_0(\mathsf{A}_0(\langle \mathcal{O}({f}_i),\mathcal{O}({E}_{k,f_i,\mu_i})\rangle_{i\in [d_0]}))$. Critically, the algorithm $\mathsf{A}_0$ is oblivious of $S$ and $\mathsf{\mathsf{Or}}_0$ has access to only a single evaluation of $S$, namely $S(\hat{v_0}\|\hat{f}_0(\hat{v}_0))$. 

The purpose of this description is to highlight that this entire procedure is essentially performed by an algorithm with access to only a single evaluation of $S$. By Theorem \ref{Raz 2}, this algorithm cannot guess $S(x)$ on any input $x\neq \hat{v_0}\|\hat{f}_0(\hat{v}_0)$. By the security of one-time programs, this algorithm can be simulated with a simulator that has access to only a single query to each program $(E_{k,f_i,\mu_i})_{i\in [d_0]}$ instead of $(\ot(E_{k,f_i,\mu_i}))_{i\in [d_0]}$. The query to $E_{k,f_i,\mu_i}$ will yield $\perp$ except with negligible probability unless $\hat{v}_0\|{f}_i(\hat{v}_0) = \hat{v}_0\|\hat{f}_0(\hat{v}_0)$ since the simulator cannot guess $S(x)$ on $x\neq \hat{v}_0\|\hat{f}_0(\hat{v}_0)$. If the condition $\hat{v}_0\|{f}_i(\hat{v}_0) = \hat{v}_0\|\hat{f}_0(\hat{v}_0)$ is not satisfied for any $i\in [d_0]$ then the simulator will receive $\perp$ from all its oracle queries except with negligible probability and thus there is a negligible probability that $\hat{y}_0= S'(\hat{f}_0(\hat{v}_0))$. Note that, there is at most a single value $i$ that satisfies this condition except with negligible probability since the parameters of the functions $(f_i)_{i\in [d_0]}$ are chosen independently and at random. 

Assume this condition is satisfied only by the function $f_n$ where $n\in [d_0]$. There is negligible chance that $\hat{v}_0=v_n$ and so $\advs$ cannot distinguish $f_n(\hat{v}_0)$ or equivalently $\hat{f}_0(\hat{v}_0)$ from random. The other values used by the oracle in the decryption query are: $v$, $S(v\|\hat{f}_0(\hat{v}_0))$ and $S'(\hat{f}_0(\hat{v}_0))$. Even if $\advs$ learns all these values, it does not help in determining functions $S$ and $S'$ since $\hat{f}_0(\hat{v}_0)$ is indistinguishable from random. 

At the end of this argument, $\advs$ cannot distinguish the key $k$ from random. Notice that the only requirement to apply this argument on a decryption query is that $\adv$ cannot distinguish $k$ from random when it submits the query. Hence, this argument can be applied to all decryption queries and it can be deduced inductively that $\advs$ cannot distinguish $k$ from random when it receives the challenge ciphertext. Let $\langle \ot(f), \ot(E_{k,f,m_b})\rangle$ be the challenge ciphertext. By one-time program security, the adversary's access to these programs can be simulated with a single evaluation to each program. Hence, it is clear that the probability that the adversary guesses $b$ is upper bounded by $\frac{1}{2}+\negl$. This still holds if $k$ is later revealed since one-time programs disappear after their transmission ends by Theorem \ref{main poly}. 

\qed     \end{proof}

\subsection{Asymmetric Key Encryption}


Let $\overline\Pi$ be the symmetric key encryption scheme in
Construction~\ref{cca}, and let
$(\textsf{KeyGen},\textsf{Br},\textsf{Eval})$ be the algorithms for
program broadcast. Use the same field $\mathbb F=\mathbb F_{2^m}$,
where $m=\lceil(\lg\lambda)^{3/2}\rceil$.

Use a polynomial-time implementation of $\overline{\textsf{Gen}}$,
truncating rejection sampling with negligible error if needed.
Choose a polynomially bounded dimension $d=d(\lambda)$ such that
$dm$ bits suffice for the randomness of
$\overline{\textsf{Gen}}(1^\lambda)$. For $u\in\mathbb F^d$, let
$\kappa(u)\coloneqq\overline{\textsf{Gen}}(1^\lambda;u)$, using
the binary representation of $u$ as random bits and ignoring any
unused bits.

Let $Q$ bound the number of evaluations available to the adversary
from the broadcast. Since the program below has output length
$(d+2)m$, the parameters of Theorem~\ref{broadcast security} give
$Q=\max\{1,\lceil\texttt{s}/(2(d+2)m)\rceil\}$.

\begin{construct}[Asymmetric Key Encryption]
\label{construct:sparse-label-public-key}
{\small
The scheme consists of the following algorithms:

\begin{itemize}
\item $\textsf{Gen}(1^\lambda,\texttt{s})$:
Choose independent uniform polynomials
$G_1,G_2,K_1,\ldots,K_d\in\mathbb F[X]$ of degree at most $Q+1$.
Define
\begin{align*}
 G(x)&\coloneqq(G_1(x),G_2(x)),
 &K(x)&\coloneqq(K_1(x),\ldots,K_d(x)),\\
 P(x)&\coloneqq G(x)\|K(x).
\end{align*}
Generate $ek\leftarrow\textsf{KeyGen}(1^\lambda,t_{\textnormal{end}})$.
The private key is $sk=(ek,G_1,G_2,K_1,\ldots,K_d)$.

\item $\textsf{KeySend}(\texttt{s},sk)$:
Before $t_{\textnormal{end}}$, output
$O_P\leftarrow\textsf{Br}(\texttt{s},ek,P)$.
At the end of the broadcast, apply the qmemory bound and release $ek$.
After this time, return $ek$.
Write $\rho_{pk}=\langle O_P,ek\rangle$.

\item $\textsf{KeyReceive}(\rho_{pk})$:
Choose $v\in_R\mathbb F$ and evaluate the public key to obtain
$G(v)\|K(v)\leftarrow\textsf{Eval}(\rho_{pk},v)$.
Output $k_v=(G(v),\kappa(K(v)))$.

\item $\textsf{Enc}(\texttt{s},k_v,\mu)$:
For $k_v=(y,k)$, output
$\langle y,\overline{\textsf{Enc}}(\texttt{s},k,\mu)\rangle$.

\item $\textsf{Dec}(sk,\langle y,\rho_{ct}\rangle)$:
Reject unless $y=(y_1,y_2)\in\mathbb F^2$. Compute
\begin{align*}
 H(X)\coloneqq\gcd\bigl(G_1(X)-y_1,G_2(X)-y_2,X^{2^m}-X\bigr).
\end{align*}
If $\deg H\ne1$, output $\perp$. Otherwise, let $v$ be its unique
root and output $\overline{\textsf{Dec}}(\kappa(K(v)),\rho_{ct})$.
\end{itemize}}
\end{construct}

We use two field elements for $G(v)$ so that its image contains at
most a $2^{-m}$ fraction of $\mathbb F^2$. As shown in the proof,
this makes a new valid value hard to guess and makes collisions at
the honest user's input unlikely. The dimension $d$ of $K(v)$
is chosen only to provide enough randomness to generate the entire
symmetric key.

The gcd computation is an efficient way to invert $y=G(v)$: its
roots are exactly the $v\in\mathbb F$ satisfying both
$G_1(v)=y_1$ and $G_2(v)=y_2$. Since $G$ need not be injective,
we also check that $y$ has exactly one preimage. Including
$X^{2^m}-X$ restricts the roots to $\mathbb F$ and removes
multiplicities. To compute $H$ efficiently, first compute
$h=\gcd(G_1-y_1,G_2-y_2)$. Reject if both polynomials are zero.
Otherwise compute $X^{2^m}\bmod h$ by $m$ repeated squarings, and
then compute its difference with $X$ and its gcd with $h$.
All these operations take time polynomial in $Q,m$.
The honest user requires $12(d+2)m=O(\lg^{3}\lambda)$ qmemory to process the broadcast, since $d=O(m)$ in our symmetric encryption scheme.
As before, pad the circuits as needed so that the broadcast error
is negligible in $\lambda$.

We use the assumed security of program broadcast in the full security
experiment, including the encryption and decryption queries and the
final revelation of the private key.

\begin{theorem}[Security]
\label{thm:sparse-label-public-key}
Construction~\ref{construct:sparse-label-public-key} is correct and
satisfies qDCCA1 security against computationally unbounded adversaries
with qmemory bound $\texttt{s}$.
\end{theorem}

\begin{proof}
Let $N\coloneqq2^m$. By polynomial interpolation, evaluations of each
chosen polynomial at any $Q+2$ distinct inputs are independent and
uniform. This also holds when the inputs are chosen adaptively from
previous evaluations. Since $Q$ is polynomial in $\lambda$, we have
$Q+2<N$ for all sufficiently large $\lambda$.

For correctness, fix the honest user's input $v$. For each $x\ne v$,
$\Pr[G(x)=G(v)]=N^{-2}$. Thus $G(v)$ has another preimage with
probability at most $(N-1)/N^2$. Otherwise, the gcd in decryption
has the unique root $v$, since $X^N-X$ has every element of
$\mathbb F$ as a simple root. Correctness follows from that of the
broadcast and symmetric encryption.

By program-broadcast security, the adversary can be simulated using
at most $Q$ classical queries to $P$. We analyze this simulation
until the private key is revealed. Let $v$ be the honest user's
input and $y^*=G(v)$. We may give $y^*$ to the adversary at the
start. After distinct queries at $x_1,\ldots,x_t$, put
$T=\{v,x_1,\ldots,x_t\}$.

Consider the following modification. Reject decryption queries with
$y\notin G(T)$. For $y=G(x)$ with $x\in T$, decrypt directly with
$\kappa(K(x))$. Stop if a query to $P$ uses $v$ or if distinct
inputs in $T$ have equal values under $G$. We first show that these
changes have negligible effect. All probabilities below are computed
in this modified experiment and bound the first different answer;
we do not condition on earlier unsuccessful guesses.

Fix the queries and answers so far, $y^*$, and even $K(v)$.
Every possible $v$ outside the queried inputs has the same conditional
probability, because the $t+1$ evaluations have the same distribution
at every choice of distinct inputs. The encryption and decryption
answers in the modified experiment depend only on these values.
Hence $v$ remains uniform among the $N-t$ unqueried inputs. The
probability that a broadcast query uses $v$ before the private key
is revealed is therefore at most $Q/(N-Q)$.

We now show that rejecting decryption queries with $y\notin G(T)$
and answering the remaining queries using the corresponding key
changes the experiment by only a negligible amount. Fix $v$. For each $x\notin T$, $G(x)$ remains uniform
in $\mathbb F^2$, even given all the evaluations at $T$. Thus, for
any decryption query with $y\notin G(T)$,
\begin{align*}
 \Pr[y\in G(\mathbb F)]
 \leq\sum_{x\notin T}\Pr[G(x)=y]\leq N^{-1}.
\end{align*}
For $q_D$ decryption queries, replacing these answers by $\perp$
costs at most $q_D/N$ in trace distance. 

Finally, equal values at distinct inputs in $T$ occur with probability
$O(Q^2/N^2)$. The probability that any value in $G(T)$ has a preimage
outside $T$ is at most $(Q+1)/N$, by the same calculation above.
Outside these events, direct decryption with the corresponding key
agrees with the gcd computation. The total trace distance between the original experiment and this modified one is therefore:
\begin{align*}
 O\left(\frac{q_D+Q+1}{N}+\frac{Q^2}{N^2}\right).
\end{align*}

We now reduce the modified experiment to symmetric encryption. The experiment now interacts with the challenger in the symmetric encryption security game. 

Choose independent uniform $v\in\mathbb F$ and
$y^*\in\mathbb F^2$. Run the adversary and answer each new query to $P$ at $x\ne v$ with
independent uniform values for $G(x)$ and $K(x)$, and reuse these
values for repeated queries. For each encryption query $\mu$,
submit $\mu$ to the encryption oracle for $\overline\Pi$ and
return $\langle y^*,\rho_{ct}\rangle$, where $\rho_{ct}$ is the
ciphertext received from that oracle. Submit the adversary's
challenge messages $\mu_0,\mu_1$ to the challenger for
$\overline\Pi$. On receiving its challenge ciphertext
$\rho_{ct}^*$, give $\langle y^*,\rho_{ct}^*\rangle$ to the
adversary. For each decryption query
$\langle y^*,\rho_{ct}\rangle$, submit $\rho_{ct}$ to the
decryption oracle for $\overline\Pi$ and return its answer.
For decryption at $G(x_i)$, run $\overline{\textsf{Dec}}$ with the known key
$\kappa(K(x_i))$. Return $\perp$ for other values. Stop on the
negligible events just considered.

Since $K(v)$ is uniform in $\mathbb F^d$, the key $\kappa(K(v))$
has the distribution of $\overline{\textsf{Gen}}(1^\lambda)$.
The other $Q$ evaluations
remain uniform and independent. Interpolation therefore gives the
correct distribution of all answers above, including adaptively
chosen queries.

When the challenger for $\overline\Pi$ reveals the key $k$,
sample $u\in\mathbb F^d$ uniformly conditional on $\kappa(u)=k$
and set $K(v)=u$. This sampling can be done by exhaustive search,
since the reduction is computationally unbounded.
Sample each polynomial uniformly subject to all evaluations already
specified, including $G(v)=y^*$ and $K(v)=u$.
There are at most $Q+1$ specified inputs, so interpolation leaves
at least one free coefficient. Reveal the resulting private key
and answer subsequent queries using these polynomials.

Let $\epsilon_{\rm Br}$ be the broadcast simulation error and
$\epsilon_{\rm sym}$ the symmetric advantage, including any
negligible error from truncating key generation. By the triangle inequality, the adversary's
success probability is at most
\begin{align*}
 \frac12+\epsilon_{\rm Br}(\lambda)+\epsilon_{\rm sym}(\lambda)
 +O\left(\frac{q_D+Q+1}{2^m}+\frac{Q^2}{2^{2m}}\right)
 =\frac12+\mathsf{negl}(\lambda).
\end{align*}
The reduction uses the original decryption access before the
challenge and reveals the full private key afterwards. 
\qed     \end{proof}

\section{Authentication Schemes}
\label{sec:authentication}

In this section, we present a signature scheme satisfying information-theoretic disappearing unforgeability under chosen message and verification attacks in the \textsf{BQSM}. 

\subsection{Message Authentication Codes}
We first present a message-authentication scheme (\textsf{MAC}) with qDCMVA security in the \textsf{BQSM}. 

\begin{construct}
\label{con:MAC}
{\small Let $\overline{\Pi}=(\overline{\textsf{Gen}}, \overline{\textsf{Enc}}, \overline{\textsf{Dec}})$ be the algorithms given in Construction \ref{cca} for symmetric key encryption. The \textsf{MAC} scheme $\Pi_{\textsf{MAC}}$ against adversaries with $\texttt{s}$ qmemory bound consists of the following algorithms. 
\begin{itemize}
  \item $\textsf{Gen}(1^\lambda)$: Output $k\leftarrow \overline{\textsf{Gen}}(1^\lambda)$.
  \item $\textsf{MAC}(\texttt{s}, k, \mu)$: Output $\rho_{\sigma}\leftarrow \overline{\textsf{Enc}}(\texttt{s}, k, \mu )$.
\item $\textsf{Verify}(k, \mu, \rho_{\sigma})$: Run $\overline{\textsf{Dec}}( k, \rho_{\sigma} )$. If the result is $\mu$ then output 1 and output 0 otherwise. 
\end{itemize}}\end{construct}

\begin{theorem}[Security]
\label{construction 4}
Construction \ref{con:MAC} ($\Pi_{\textsf{MAC}}$) satisfies qDCMVA security against computationally unbounded adversaries with $\texttt{s}$ qmemory bound. 
\end{theorem}

\begin{proof}
{The algorithms $({\textsf{Gen}}, {\textsf{MAC}}, {\textsf{Verify}})$ are built from the algorithms $(\overline{\textsf{Gen}}, \overline{\textsf{Enc}}, \overline{\textsf{Dec}})$ in Construction \ref{cca}.}

{In the qDCMVA experiment, a key is sampled $k\leftarrow \textsf{Gen}(1^\lambda)$ and the adversary is run $\advs^{\textsf{MAC}(k,\cdot),\textsf{Verify}(k,\cdot)}$.  }

{Denote the tags produced by all queries to $\textsf{MAC}(k,\cdot)$ as $\langle \mathcal{O}({f}_i),\mathcal{O}({E}_{k,f_i,\mu_i})\rangle_{i\in [Q]}$ and let the forgery produced by the adversary after these queries be $\langle \mathcal{O}(\hat{f}),\mathcal{O}(\hat{E})\rangle \leftarrow \advs^{\textsf{Verify}(k,\cdot)}$ for $\mathbf{NC}^1$ functions $\hat{f}:\{0,1\}^m \rightarrow \{0,1\}^m$ and $\hat{E}:\{0,1\}^{2m}\rightarrow \{0,1\}^{m+\ell}$. If the output is not of this form, it is immediately rejected as an invalid tag. }

{The experiment then runs the verification algorithm on the forgery. Specifically, it evaluates the first program on a random input, say $\hat{v}$, and obtains $\hat{f}(\hat{v})$. Then, it evaluates the second program on $S(\hat{v}\|\hat{f}(\hat{v}))$. We argued in the proof of Theorem \ref{construction 1} that, except with negligible probability, we must have $\hat{v}\|\hat{f}(\hat{v})=\hat{v}\|f_i(\hat{v})$ for some $i\in [Q]$. }

{Note that prior to sending $\ot(\hat{E})$, the memory bound for $(\ot(f_i))_{i\in [Q]}$ must have already been applied; thus there exists values $(v_i)_{i\in [Q]}$ such that $\advs$ can be simulated with evaluations $f_i(v_i)$. In other words, $\advs$ can only distinguish $f_i(x)$ from random on $x=v_i$ for all $i\in [Q]$. This is because these functions are of the form $f_i(x)=a_ix+b_i \ \mod 2^{m}$ so one evaluation is not sufficient to distinguish the evaluation on another input from random.  }

{To sum up, to pass verification with non-negligible probability, $\advs$'s forgery must satisfy $\hat{v}\|\hat{f}(\hat{v})=\hat{v}\|f_i(\hat{v})$ for some $i\in [Q]$ and at the same time, $\advs$ cannot distinguish $f_i(x)$ on $x\neq v_i$ for any $i\in [Q]$. Therefore, since $(v_i)_{i\in [Q]}$ are chosen independently of $\hat{v}$, there is a negligible probability that both these conditions are satisfied. Hence, Construction \ref{con:MAC} is qDCMVA secure.}
\qed     \end{proof}

\subsection{Signatures}

Let $\overline\Pi=(\overline{\textsf{Gen}},\overline{\textsf{MAC}},
\overline{\textsf{Verify}})$ be a correct qDCMVA-secure MAC scheme
with classical keys of length $\ell_k=\ell_k(\lambda)$. As in the
MAC construction above, suppose its honest algorithms require no
qmemory. Let $(\textsf{KeyGen},\textsf{Br},\textsf{Eval})$ be the
algorithms for program broadcast. For polynomially bounded
$\texttt{s}$, set
\begin{align*}
 t&\coloneqq\lceil(\lg\lambda)^{3/2}\rceil,
 &n&\coloneqq2(\lceil\texttt{s}\rceil+t),
 &b&\coloneqq\ell_k+\lceil\lg n\rceil+t.
\end{align*}
The additional $t$ output bits below ensure that $b=\omega(\lg\lambda)$.
We use the same broadcast security guarantee as in the asymmetric
encryption proof, including the interaction with the signing and
verification oracles.

\begin{construct}[Signatures]
\label{con:Sig}
{\small
\begin{itemize}
The signature scheme consists of the following algorithms:
\item $\textsf{Gen}(1^\lambda,\texttt{s})$:
Generate independent keys
$k_i\leftarrow\overline{\textsf{Gen}}(1^\lambda)$ for $i\in[n]$
and choose a uniform permutation $\pi$ of $[n]$. Define the
program with $b$ output bits
\begin{align*}
 P(j)\coloneqq\pi(j)\|k_{\pi(j)}\|0^t
 \qquad\textnormal{for $j\in[n]$.}
\end{align*}
Return $\perp$ on invalid inputs, using a distinct $b$-bit encoding.
Generate $ek\leftarrow\textsf{KeyGen}(1^\lambda,t_{\textnormal{end}})$.
The private key is $sk=(ek,\pi,k_1,\ldots,k_n)$.

\item $\textsf{KeySend}(\texttt{s},sk)$:
Before $t_{\textnormal{end}}$, output $t$ transmissions
$O_1,\ldots,O_t\leftarrow\textsf{Br}(\texttt{s},ek,P)$,
using independent randomness and the same $ek$.
At the end of the broadcast, apply the qmemory bound and release
$ek$. After this time, return $ek$.
Write $\rho_{pk}=\langle O_1,\ldots,O_t,ek\rangle$.

\item $\textsf{KeyReceive}(\rho_{pk})$:
Choose $t$ distinct uniform inputs $j_1,\ldots,j_t\in[n]$.
Evaluate $O_a$ on $j_a$ using $\textsf{Eval}$ to obtain
$i_a\|k_{i_a}\|0^t=P(j_a)$ for each $a\in[t]$.
Let $T\coloneqq\{i_1,\ldots,i_t\}$.
Discard the inputs $j_1,\ldots,j_t$ and output the classical
verification key $k_T=(T,(k_i)_{i\in T})$.

\item $\textsf{Sign}(\texttt{s},sk,\mu)$:
For $i=1,\ldots,n$, generate and send
$\rho_{\sigma,i}\leftarrow
\overline{\textsf{MAC}}(\texttt{s},k_i,\mu)$ using fresh randomness.
The signature is $\rho_\sigma=\langle\rho_{\sigma,1},\ldots,
\rho_{\sigma,n}\rangle$.

\item $\textsf{Verify}(k_T,\mu,\rho_\sigma)$:
Check the format and process the components in order. Discard
component $i$ if $i\notin T$; otherwise run
$\overline{\textsf{Verify}}(k_i,\mu,\rho_{\sigma,i})$.
Perform all the checks and output $1$ at the end if every check
returns $1$; otherwise output $0$.
\end{itemize}}
\end{construct}

The position $i$ of a MAC tag is public, whereas its corresponding
broadcast input $\pi^{-1}(i)$ is secret. Thus, learning which
components a verifier checks does not reveal which inputs give
their keys. Neither signing nor verification uses $\pi$.
Append $\lambda$ ignored input bits to $P$, filled with zeros by
honest users, to satisfy the circuit-size and error requirements
of the broadcast construction.

The honest user uses $12tb$ qmemory to receive the public key.
Signature components are processed as they arrive, so signing and
verification require no qmemory. For the MAC above,
$\ell_k=O((\lg\lambda)^{3/2})$, and hence the honest memory is
$O((\lg\lambda)^3)$ for polynomial $\texttt{s}$. In particular,
choosing $\texttt{s}=\lambda^c$ for any fixed $c>0$ gives the gap
$h$ versus $2^{\Theta(h^{1/3})}$, where $h=12tb$ is the honest
qmemory. Each signature contains $n$ MAC tags.

\begin{theorem}[Security]
\label{thm:signature-from-mac}
Construction~\ref{con:Sig} is correct and satisfies
qDCMVA security against computationally unbounded adversaries with
qmemory bound $\texttt{s}$.
\end{theorem}

\begin{proof}
Correctness follows from correctness of program broadcast and the
MAC scheme. The honest user obtains $t$ distinct keys and checks
valid MAC tags under these keys. There are only polynomially many
evaluations, so their total correctness error is negligible.

By broadcast correctness, replacing the honest user's evaluations
by $P(j_1),\ldots,P(j_t)$ changes the experiment by a negligible
amount. Let $Q\coloneqq\lfloor\texttt{s}/(2b)\rfloor\leq n/2$.
By program-broadcast security, simulate all public-key transmissions
to the adversary using at most $Q$ classical queries to $P$. Let $C\subseteq[n]$ contain the values $\pi(j)$ obtained from
these queries. 

Since the honest inputs are a uniform set of size $t$, the set
$T$ is uniform among the $t$-element subsets of $[n]$, independently
of $\pi$ and all MAC keys. All signing and verification operations depend only on $T$ and the MAC keys. Even after fixing all these keys, the answers to these operations give no information about the unqueried entries of $\pi$. Consequently, a query at a previously unqueried input
$j\in[n]$ reveals a uniform element of $[n]\setminus C$, even if
the queries are chosen
adaptively.
Therefore
\begin{align*}
 \Pr[T\subseteq C]
 \leq\frac{\binom Qt}{\binom nt}
 \leq\left(\frac Qn\right)^t
 \leq2^{-t},
\end{align*}
where the first bound is zero if $Q<t$. The same bound applies
if the adversary stops early, by completing its queries to $Q$
distinct inputs for this calculation.

We reduce a forgery when $T\not\subseteq C$ to a MAC forgery. The reduction interacts with the MAC challenger as follows.

Choose uniform $T$ and $I\in_R T$, and choose $\pi$ independently.
Use the unknown key of the MAC challenger as $k_I$, and generate
all other keys locally. To answer a broadcast query at $j$, compute
$i=\pi(j)$. If $i=I$, abort; otherwise return $i\|k_i\|0^t$.
Repeated queries receive the same answer.

For each signing query $\mu$, obtain component $I$ from the MAC
oracle on $\mu$, and generate all other components with the known
keys. For each verification query, process the components in the
prescribed order, using the verification oracle for component $I$
and the known keys for the other components in $T$. Return $1$
exactly when all $t$ checks return $1$. Continue answering
verification queries after signing queries end, as required by
the original qDCMVA experiment. Since the honest MAC algorithms
use no qmemory, these local computations do not increase the
reduction's qmemory at the prescribed bounds.

Until an abort, this simulation has exactly the distribution of
the experiment above. In that experiment, $I$ can equivalently
be chosen uniformly from $T$ after the full interaction, since
it does not affect the distribution of any key or answer.
If $T\not\subseteq C$, then $I\notin C$ with probability at
least $1/t$.

For a final forgery $(\mu,\rho_\sigma)$, output
$(\mu,\rho_{\sigma,I})$ to the MAC challenger and discard the
other components. If the signature is accepted, its component $I$
is accepted by MAC verification. 

Let $p$ be the signature forgery probability after the broadcast
simulation. The reduction succeeds with probability at least
$(p-2^{-t})/t$. Thus, writing $\epsilon_{\rm Br}$ for the
total broadcast correctness and simulation error, and
$\epsilon_{\rm MAC}$ for the
success probability of the MAC adversary, we obtain
\begin{align*}
 \Pr[\textsf{Sign}^{\textsf{qDCMVA}}_{\Pi,\mathcal A_{\texttt{s}}}
       (\lambda)=1]
 \leq\epsilon_{\rm Br}(\lambda)+2^{-t}
       +t\,\epsilon_{\rm MAC}(\lambda)
 =\mathsf{negl}(\lambda).
\end{align*}
This proves security in the original experiment, including a
forgery of a previously signed message.
\qed     \end{proof}

\section{Discussion: Asymmetric Key Cryptography}
\label{sec:discussion}

It might seem that our asymmetric schemes are unorthodox because each copy of the public key differs as a result of the random choices made during $\textsf{KeySend}$. A change of perspective through a standard purification argument shows that this is not really the case. Let $r$ be a string representing all the random choices made when sending a copy of the public key denoted as $\ket{pk_r}$. Instead, the sender can send the second register of $\sum_r\sqrt{p_r} \ket{r} \otimes \ket{pk_r}$, then measure the first register and send the outcome to the receiver which at the end gives the same public key as in our schemes. Such keys are called mixed-state and we argue in the introduction that there does not seem to be any advantages for pure-state keys over mixed-state keys. 

Another topic of discussion is the qmemory requirements of our asymmetric protocols. Our one-time compiler, private key encryption, and message-authentication schemes do not require any qmemory for the honest user. This is highly desirable since it means the protocols are not difficult to implement with current technology. In comparison, our asymmetric key schemes require users to store quantum states.


We now show that this limitation is unavoidable in any asymmetric key scheme covered in this work. This includes any scheme with classical, pure-state, or mixed-state public keys and covers all signature and asymmetric key encryption schemes. We say a scheme is \emph{$m$-qmemory} if it requires at most $m$ qmemory for all honest parties.  

\begin{theorem}[Impossibility of 0-qmemory Asymmetric Protocols]
\label{0 memory}
There does not exist a correct 0-qmemory asymmetric encryption
or signature scheme with reusable classical receiver keys that
is secure against computationally unbounded adversaries given
an arbitrary polynomial number of independently generated public
key copies, even when the adversary has no qmemory.
\end{theorem}

Note that Theorem \ref{0 memory} applies to both digital signatures and asymmetric encryption. We now give a proof of Theorem \ref{0 memory}.

\begin{proof}
Assume $\Pi$ is a 0-qmemory asymmetric scheme as defined above.
The private key $sk$ and the key $k$ obtained by
$\textsf{KeyReceive}$ are classical. Let $p=p(\lambda)$ be a
polynomial bound on the private key length, padding keys if
necessary. For a fixed $sk$, denote by $D_{sk}$ the distribution of
\begin{align*}
 k\leftarrow\textsf{KeyReceive}(\textsf{KeySend}(sk)),
\end{align*}
including all randomness and measurement outcomes.

For encryption, fix two distinct messages $\mu_0,\mu_1$ of the
same length. For each candidate private key $sk'\in\{0,1\}^p$,
define
\begin{align*}
 a_{sk'}(k)\coloneqq\frac12\sum_{b\in\{0,1\}}
 \Pr[\textsf{Dec}(sk',\textsf{Enc}(k,\mu_b))=\mu_b].
\end{align*}
For signatures, fix a message $\mu$ and instead define
\begin{align*}
 a_{sk'}(k)\coloneqq
 \Pr[\textsf{Verify}(k,\mu,\textsf{Sign}(sk',\mu))=1].
\end{align*}
In either case, $a_{sk'}(k)\in[0,1]$. Since the algorithms are
known and the adversary is computationally unbounded, it can
approximate these probabilities by classical simulation.

Set $\epsilon=1/100$. The adversary runs $\textsf{KeyReceive}$ on
$n$ independently generated public key copies and stores the
classical outputs $k_1,\ldots,k_n$.
Define
\begin{align*}
 \widehat a(sk')\coloneqq
 \frac1n\sum_{i=1}^n a_{sk'}(k_i).
\end{align*}

The adversary can approximate $\widehat a$ on any key up to $\epsilon$ error. It chooses $sk^*$ maximizing its approximation. Let
\begin{align*}
 A(sk')\coloneqq\mathbb E_{k\leftarrow D_{sk}}[a_{sk'}(k)].
\end{align*}
Conditioned on $sk$, the samples are independent. By Hoeffding's
inequality and a union bound over the $2^p$ candidate keys,
\begin{align*}
 \Pr\left[\exists sk':
 \left|\frac1n\sum_{i=1}^n a_{sk'}(k_i)-A(sk')\right|
 >\epsilon\right]
 \leq 2^{p+1}e^{-2n\epsilon^2}.
\end{align*}
Choose $n=\lceil(p+\lambda+1)\ln(2)/(2\epsilon^2)\rceil$,
which is polynomial in $\lambda$.

Except with probability $2^{-\lambda}$, every sample average
$\widehat a(sk')$ is within $\epsilon$ of $A(sk')$.
Since the adversary chooses $sk^*$ using approximations with
error at most $\epsilon$, we obtain
\begin{align*}
 A(sk^*) &\geq \widehat a(sk^*)-\epsilon\\
         &\geq \widehat a(sk)-3\epsilon\\
         &\geq A(sk)-4\epsilon.
\end{align*}

The honest receiver's key is an independent sample from $D_{sk}$,
conditioned on $sk$. For encryption, the adversary decrypts the
challenge using $sk^*$ and outputs the corresponding bit, guessing
arbitrarily if neither message is returned. Its success probability
is at least $A(sk^*)$. For signatures, it outputs
$\textsf{Sign}(sk^*,\mu)$ without making any signing queries, and
verification accepts with probability $A(sk^*)$. Averaging over
$sk$ and the samples, correctness therefore gives success
probability at least
\begin{align*}
 1-4\epsilon-2^{-\lambda}-\mathsf{negl}(\lambda)
\end{align*}
in either case. All computation before the attack is classical
apart from running $\textsf{KeyReceive}$, and the attack itself
uses the honest decryption or signing algorithm. These algorithms
require no qmemory. Thus, a computationally unbounded adversary
with no qmemory breaks security using only polynomially many
public key copies.
\qed     \end{proof}

Note that the memory bound of an adversary should certainly be polynomial with respect to the security parameter $\lambda$ i.e. $\texttt{s}\in \poly[\lambda]$. This is because all states in an efficient scheme must be polynomial in length and an adversary with super-polynomial memory can store an unbounded (polynomial) number of states. Such an adversary is essentially not operating in the \textsf{BQSM}. It is well known that encryption and authentication are impossible unconditionally in the standard model. This implies the required qmemory of the honest user in our schemes may be improved slightly but the qmemory required to break the schemes is optimal. 

Many quantum protocols, such as quantum key distribution and oblivious transfer, do not require any qmemory for honest users. Theorem \ref{0 memory} is interesting since it shows that further possibilities can be realized using qmemory. In fact, to our knowledge, this is the first primitive in the \textsf{BQSM} that requires qmemory for the honest user. 

\section{Note on Noisy Communication}
\label{sec:Noisy Communication}

Our protocols assume that the communication between the sender and receiver is error-free. It is assumed that the sender has a perfect quantum source which when requested will produce one and only one qubit in the correct state. Note if the source accidentally produces two copies of a qubit then the adversary can measure the qubit in two bases and break the 1-2 oblivious transfer scheme. Furthermore, it requires the honest evaluator to measure all qubits without error. 

Fortunately, in \cite{dfrss07} the authors used techniques based on information reconciliation to allow an honest evaluator and sender in the quantum $1\text{-}2$ oblivious transfer scheme (Construction \ref{con:QOT}) with imperfect apparatus to perform error correction without compromising the security of the scheme. The same techniques are applicable to our protocols since the quantum parts of our schemes consist of oblivious transfer transmissions. However, we leave a more rigorous treatment of these issues as an avenue for future work.


\section*{Appendices to Chapter 1}

\renewcommand{\thesection}{\thechapter.\Alph{section}}
\setcounter{section}{0}

\section{One-Time Programs}
\label{sec:one-time construction}
This section presents an information-theoretically secure \textsf{BQS} one-time compiler. We now give a brief description of his construction. To start, we build one-time programs for a class of functions known as matrix branching programs (see Sec.~\ref{sec:branching}). These programs are important because any circuit in $\textbf{NC}^1$ can be represented as a matrix branching program by Barrington's Theorem \cite{B89}. The first step is to randomize the program by left and right multiplication of each instruction with random matrices so that the randomization cancels out when performing an evaluation. Next, the program is entirely one-time padded so that everything is hidden and this encrypted program is sent to the evaluator. Oblivious transfer is employed to securely send 1-out-of-2 one-time pad keys for each pair of instructions in the program. This reveals only one consistent randomized \emph{path} in a branching program which is sufficient to learn a single evaluation and nothing else. 

To generalize this to polynomial circuits, the circuit is sliced into a sequence of subcircuits in $\textbf{NC}^1$ and then a one-time program for each subcircuit is sent. However, this opens the avenue for various attacks: an adversary can learn information about the individual subcircuits, perform an evaluation with inconsistent input, and so forth. To protect against these attacks, all the subcircuit outputs are padded and a one-time decryption circuit is sent after the evaluation. This decryption circuit, in $\textbf{NC}^1$, checks if the evaluation is performed correctly and, if so, reveals the padding of the final subcircuit. 

\subsection{One-time Compiler for Branching Programs}
\label{sec:protocol}
Let $\mathcal{P}=(M_j^b)_{j\in [N],b\in\{0,1\}}$ be a branching program on input $\{0,1\}^n$, of length $N$ (polynomial in $n$), and width $k=5$. Two layers of security are used to prepare the one-time program of $\mathcal{P}$. First, the compiler randomly chooses a set of matrices $(D_{j})_{j\in [N-1]}$, where $D_j\in_R \hil{M}_5$. The compiler then modifies the program as follows: 
\begin{align*}
\label{first layer}
(M_{j}^b)_{j\in [N],b\in \{0,1\}}\rightarrow (D^{-1}_{j-1}M_{j}^bD_{j})_{j\in [N],b\in \{0,1\}} ,
\end{align*}
where $D_{-1}$ and $D_{N-1}$ are defined to be the identity. Clearly this program has the same functionality as $\mathcal{P}$ since the $D_{j}$ matrices cancel out during composition. Roughly speaking, this encoding is only sufficient to hide a single path in the program (see Lemma \ref{path}). The next step is to add a second layer of security by one-time padding the entire program. To do this the compiler chooses $2N$ matrices $(B^b_{j})_{j\in[N],b\in \{0,1\}}$, where $B^b_j\in_R \hil{M}_5$, then modifies the program as follows: 
\begin{align*} 
(D^{-1}_{j-1}M_{j}^bD_{j})_{j\in [N],b\in \{0,1\}} \rightarrow (B^b_{j}D^{-1}_{j-1}M_{j}^bD_{j})_{j\in [N],b\in\{0,1\}} .
\end{align*}
We denote this encrypted program as $E(\mathcal{P})$. Note this modification completely destroys the functionality and information in the program. To allow the evaluator to get an evaluation of $\mathcal{P}$ we use oblivious transfer as follows. 

Let the string $s^b_i$ encode the matrices $(B^b_j)_{j=i\bmod n}\coloneqq (B^b_i,B^b_{i+n},...,B^b_{N-n+i})$. Let $\rho_{m,\ell}(s^0,s^1)$ denote the state and information sent in a $\otqot{m, \ell}(s^0,s^1)$ transmission (Construction \ref{con:QOT}). Hence, transmitting $\bigotimes_{i=1}^n\rho_{m,\ell}(s^0_i,s_i^1)$ allows a receiver to uncover a consistent path in the encrypted branching program $E(\mathcal{P})$. 

There are $\frac{N}{n}$ matrices in $(B^b_{j})_{j= i\bmod n}$ and each is of dimension 5 which means each requires $25$ bits to encode. Hence, $\ell = \lvert s_i^b\rvert = \frac{25N}{n}$. By Theorem \ref{QOT}, to obtain security against any adversary $\mathcal{A}_{\texttt{s}}$, $m$ must be large enough so that $\frac{m}{4}-2\ell -\texttt{s} \in \Omega (m)$. It must also be large enough to satisfy the security condition of Definition \ref{def:BQS one-time}. All in all, a value that meets these conditions is $m=250N+4\texttt{s}$. 

\begin{construct}[\textsf{BQS} One-Time Branching Program Compiler $\mathcal{O}_{\texttt{s}}$]
\label{con:branching program}
{\small Let $\mathcal{P}$ be a branching program of length $N$ on permutation matrices of dimension $k=5$ and input dimension $n$. Let $m=250N+4\texttt{s}$ and $\ell=\frac{25N}{n}$.
\begin{enumerate}
\item $\mathcal{O}_{\texttt{s}}$ randomly picks matrices $(D_{j})_{j\in [N-1]}$ and $(B^b_{j})_{j\in[N],b\in \{0,1\}}$.
\item $\mathcal{O}_{\texttt{s}}$ prepares and sends the state:
\begin{align*}|OT\rangle_n=\bigotimes_{i=1}^n\rho_{m,\ell} (s_i^{0},s_i^{1})\end{align*}
based on the choice of $(B^b_{j})_{j\in[N],b\in \{0,1\}}$. 
\item $\mathcal{O}_{\texttt{s}}$ sends $E(\mathcal{P})$ classically.
\end{enumerate}
}\end{construct}

In the next subsection, we prove that protocol $\mathcal{O}_{\texttt{s}}$ is an efficient $(0,\texttt{s})$-\textsf{BQS} one-time compiler for matrix branching programs. In order to do this, there are 5 conditions that need to be met regarding functionality, efficiency, and security according to Definition \ref{def:BQS one-time}. 

\subsection{Functionality and Efficiency}

The protocol $\mathcal{O}_{\texttt{s}}$ can easily be seen to preserve functionality efficiently.

\begin{lemma}[Functionality \& Polynomial Slowdown]
\label{functionality}
For any program $\mathcal{P}$ and input $w$, an evaluator $\mathcal{E}$ can obtain $\mathcal{P}(w)$ from $\mathcal{O}_{\texttt{s}}(\mathcal{P})$. The one-time program has a polynomial slowdown and both the sender and evaluator require no qmemory. 
\end{lemma}

\begin{proof}
The evaluator $\mathcal{E}$ can follow the oblivious transfer protocol described in Construction \ref{con:QOT} to gain access to the strings $(s^{w_i}_i)_{i\in [n]}$ from the state $|OT\rangle_n$ in Step 2 of the protocol. These can be used to determine the matrices $(B^{w_{(j\bmod n)}}_{j})_{j\in [N]}$. Furthermore, $\mathcal{E}$ has $E(\mathcal{P})$ from Step 3 in the protocol. Hence, the evaluator can invert the matrices $B^{w_{(j\bmod n)}}_{j}$ and obtain $(D^{-1}_{j-1}M_{j}^{w_{(j\bmod n)}}D_{j})_{j\in [N]}$.
Next, $\mathcal{E}$ can compose these matrices to obtain 
\begin{align*}M_{0}^{w_{0}}D_{0}\circ...\circ D^{-1}_{N-2}M_{N-1}^{w_{n-1}}=M^{w_0}_{0}\circ ...\circ M^{w_{n-1}}_{N-1}.\end{align*}
If the result is $Q_{rej}$ then $\mathcal{P}(w)=0$ and if the result is $Q_{acc}$ then $\mathcal{P}(w)=1$ and undefined otherwise.

Preparing $|OT\rangle_n$ requires $nO(N+\texttt{s})$ space and time. Next, preparing $E(\mathcal{P})$ requires $O(N)$ space and time. Hence, $\lvert \mathcal{O}_{\texttt{s}}(\mathcal{P})\rvert \in O(n\texttt{s}+Nn)$.

It is clear that both the sender and evaluator require no qmemory since Construction \ref{con:QOT} for oblivious transfer requires no qmemory by either party. 

\qed     \end{proof}

\subsection{Security}
\label{sec:security}

Before proving security, we mention the following lemma which states that a single path in a branching program randomized with Kilian's method reveals nothing except for the evaluation of the program on the input corresponding to that path. 

\begin{lemma}[Kilian Randomization \cite{k88}]
\label{path}
For any two sets $\{M_{j}\}_{j\in[N]}$ and $\{M_{j}'\}_{j\in[N]}$ of matrices such that $M_j,M_j'\in \hil{M}_k$ and $\prod_{j\in N}M_j=\prod_{j\in N}M'_j$, there exists matrices $(D_{j})_{j\in [N-1]}$ such that $D_j\in \hil{M}_k$ and
\begin{align*}
\prod_{j\in N}M_j'=\prod_{j\in N}D_{j-1}^{-1}M_jD_j .
\end{align*}
\end{lemma}

\textit{Note: Due to this randomization, a simulator can simulate a path in a Kilian randomized program $(D_{j-1}^{-1}M^{w_{j \bmod n}}_jD_j)_{j\in [N]}$ by simply querying the oracle on $w$ and choosing random permutation matrices $M_j'$ whose composition is equal to $\prod_{j\in N}D_{j-1}^{-1}M^{w_{j \bmod n}}_jD_j$.}

Now we are ready to prove security. 

\begin{theorem}[One-Time Branching Programs]
\label{main}
Let $n,\texttt{s} \in \mathbb{N}$ and let $(P_i)_{i\in [p]}$ be a set of branching programs on domain $\{0,1\}^n$ and polynomial (in $n$) lengths $(N_i)_{i\in [p]}$ of width 5. Then for any adversary $\adv_{\texttt{s}}$, there exists a simulator $\mathcal{S}_{\texttt{s}}$ such that, 
\begin{align*}|Pr[\adv((\mathcal{O}_{\texttt{s}}(P_i))_{i\in[p]})=1]-Pr[ \mathcal{S}_{\texttt{s}}^{1\mathcal{P}_0,...,1\mathcal{P}_{p-1}}(|0\rangle^{\otimes N_0},...,|0\rangle^{\otimes N_{p-1}})=1]\rvert \leq \negl[n].\end{align*}
\end{theorem}

\begin{proof}
Without loss of generality, assume the length of all the programs is $N$ and let $m=250N+4\texttt{s}$. All matrices in the proof are in $\hil{M}_5$. The entire interaction between the adversary and the one-time programs can be summarized as follows.

\smallskip \noindent\fbox{%
  \parbox{\textwidth}{%
\textbf{$\mathcal{A}_{\texttt{s}}(\mathcal{O}_{\texttt{s}}(\mathcal{P}_0),...,\mathcal{O}_{\texttt{s}}(\mathcal{P}_{p-1}))$}

{\small
\begin{enumerate}
\item $\mathcal{A}_{\texttt{s}}$ is run. 

\item $\mathcal{A}_{\texttt{s}}$ requests one-time programs of $\mathcal{P}_0,\mathcal{P}_2,...,\mathcal{P}_{p-1}$. 

\item For $k\in [n]$, $i\in [p]$, and $j\in [N]$, $\mathcal{O}_{\texttt{s}}$ randomly chooses ${x}_{k,i}\in_R\{0,1\}^{m}$, ${\theta}_{k,i} \in_R\{+,\times \}^{m}$, and ${f}_{k,i,b}\in_R \textsf{H}$.

\item $\mathcal{O}_{\texttt{s}}$ sends $|R\rangle$ to $\mathcal{A}_{\texttt{s}}$ where
\begin{align*}|{R}\rangle \coloneqq \bigotimes_{i=1}^p\bigotimes_{k=1}^n |{x}_{k,i}^1\rangle_{{\theta}_{k,i}^1},|{x}_{k,i}^2\rangle_{{\theta}_{k,i}^2},...,|{x}_{k,i}^{m}\rangle_{{\theta}_{k,i}^{m}}.\end{align*}

\item[*] \texttt{$\adv_{\texttt{s}}$'s memory bound $\texttt{s}$ applies.}

\item $\mathcal{O}_{\texttt{s}}$ chooses random matrices $(B_{j,i}^b)_{j,i,b}$ and $(D_{j,i})_{j,i}$ and correspondingly determines the appropriate strings $(s_{k,i,b})_{k,i,b}$. 

\item $\mathcal{O}_{\texttt{s}}$ sends ${\theta}_{k,i}$, ${f}_{k,i,b}$ and ${e}_{k,i,b}\coloneqq {f}_{k,i,b}({x}_{k,i}|_{{I}_{k,i,b}}) \oplus {s}_{k,i, b}$ where ${I}_{k,i,b}\coloneqq \{k:{\theta}_{k,i}=[+,\times]_b\}$.

\item $\mathcal{O}_{\texttt{s}}$ sends the encoded programs $(B^b_{j,i}D^{-1}_{j-1,i}M^b_{j,i}D_{j,i})_{j\in [N],i\in [p]}$ to $\mathcal{A}_{\texttt{s}}$. 
\end{enumerate} }
  }}
\smallskip

The simulator's algorithm is as follows. 

\smallskip \noindent\fbox{%
  \parbox{\textwidth}{%
\textbf{Simulator $\mathcal{S}_{\texttt{s}}^{1\mathcal{P}_0,...,1\mathcal{P}_{p-1}}(|0\rangle^{\otimes N})$}

{\small
\begin{enumerate}
\item $\mathcal{A}_{\texttt{s}}$ is run. 

\item $\mathcal{A}_{\texttt{s}}$ requests one-time programs for $\mathcal{P}_0,...,\mathcal{P}_{p-1}$. 

\item For $k\in [n]$, $j\in [N]$, and $i\in [p]$, $\mathcal{S}_{\texttt{s}}$ randomly chooses $\tilde{x}_{k,i}\in_R\{0,1\}^{m}$, $\tilde{\theta}_{k,i} \in_R\{+,\times \}^{m}$, and $\tilde{f}_{k,i,b}\in_R \textsf{H}$.

\item $\mathcal{S}_{\texttt{s}}$ sends to $\mathcal{A}_{\texttt{s}}$:
\begin{align*}|\tilde{R}\rangle \coloneqq \bigotimes_{i=1}^p\bigotimes_{k=1}^n |\tilde{x}_{k,i}^1\rangle_{\tilde{\theta}_{k,i}^1},|\tilde{x}_{k,i}^2\rangle_{\tilde{\theta}_{k,i}^2},...,|\tilde{x}_{k,i}^{m}\rangle_{\tilde{\theta}_{k,i}^{m}}.\end{align*}

\item[*] \texttt{$\adv_{\texttt{s}}$'s memory bound $\texttt{s}$ applies.}

\item $\mathcal{S}_{\texttt{s}}$ analyses the action of $\mathcal{A}_{\texttt{s}}$ on $|\tilde{R}\rangle$ and determines the selection bits $( \tilde{c}_{k,i})_{k,i}$. 

\item For $i\in [p]$, $\mathcal{S}_{\texttt{s}}$ queries the oracle to learn $\mathcal{P}_i(\tilde{c}_{0,i}...\tilde{c}_{n,i})$. Let $\tilde{P}_i\coloneqq (\tilde{M}_{j,i})_{j\in [N],b\in\{0,1\}}$ be a constant branching program mapping all values to $ \mathcal{P}_i(\tilde{c}_{0,i}...\tilde{c}_{n,i})$.

\item $\mathcal{S}_{\texttt{s}}$ chooses random matrices $(\tilde{B}_{j,i}^b)_{j,i,b}$ and $(\tilde{D}_{j,i})_{j,i}$ and correspondingly chooses the appropriate strings $(\tilde{s}_{k,i,b})_{k,i,b}$. 

\item $\mathcal{S}_{\texttt{s}}$ sends $\tilde{\theta}_{k,i}$, $\tilde{f}_{k,i,b}$ and $\tilde{e}_{k,i,b}\coloneqq \tilde{f}_{k,i,b}(\tilde{x}_{k,i}|_{\tilde{I}_{k,i,b}}) \oplus \tilde{s}_{k,i, b}$ where $\tilde{I}_{k,i,b}\coloneqq \{k:\tilde{\theta}_{k,i}=[+,\times]_b\}$.

\item $\mathcal{S}_{\texttt{s}}$ sends the encoded programs $(\tilde{B}^b_{j,i}\tilde{D}^{-1}_{j-1,i}\tilde{M}^b_{j,i}\tilde{D}_{j,i})_{j\in [N],i\in [p]}$ to $\mathcal{A}_{\texttt{s}}$. 

\item $\mathcal{S}_{\texttt{s}}$ outputs the output of $\mathcal{A}_{\texttt{s}}$.
\end{enumerate} }
  }}
\smallskip

The above simulation is indistinguishable with respect to the adversary to the following purified version. Note in this version, the simulator has unbounded qmemory. 

\smallskip \noindent\fbox{%
  \parbox{\textwidth}{%
\textbf{Purified Simulator $\mathcal{S}^{1\mathcal{P}_0,...,1\mathcal{P}_{p-1}}(|0\rangle^{\otimes N})$}

{\small
\begin{enumerate}
\item $\mathcal{A}_{\texttt{s}}$ is run. 

\item $\mathcal{A}_{\texttt{s}}$ requests one-time programs for $\mathcal{P}_0,...,\mathcal{P}_{p-1}$. 

\item For $w\in [m]$, $k\in [n]$, $i\in [p]$, $\mathcal{S}$ prepares an EPR pair $|\phi_{k,i,w}\rangle\coloneqq \frac{|00\rangle +|11\rangle}{\sqrt{2}}$ and sends one half to $\advs$ and stores the other half. 

\item[*] \texttt{$\adv_{\texttt{s}}$'s memory bound $\texttt{s}$ applies.}

\item For $w\in [m]$, $k\in [n]$, $i\in [p]$, $\mathcal{S}$ chooses $\tilde{\theta}_{k,i,w} \in_R\{+,\times \}$ and measures its half of $|\phi_{k,i,w}\rangle$ in basis $\tilde{\theta}_{k,i,w}$ and obtains an outcome $\tilde{x}_{k,i,w}$. 

\item $\mathcal{S}$ analyses the action of $\mathcal{A}_{\texttt{s}}$ on its halves of the EPR pairs and determines the selection bits $( \tilde{c}_{k,i})_{k,i}$. 

\item For each $i\in [p]$, $\mathcal{S}$ queries the oracle to learn $\mathcal{P}_i(\tilde{c}_{0,i}...\tilde{c}_{n,i})$. Let $\tilde{\mathcal{P}}_i\coloneqq (\tilde{M}_{j,i})_{j\in [N],b\in\{0,1\}}$ be a constant branching program mapping all values to $ \mathcal{P}_i(\tilde{c}_{0,i}...\tilde{c}_{n,i})$.

\item $\mathcal{S}_{\texttt{s}}$ chooses random matrices $(\tilde{B}_{j,i}^b)_{j,i,b}$ and $(\tilde{D}_{j,i})_{j,i}$ and correspondingly chooses the appropriate strings $(\tilde{s}_{k,i,b})_{k,i,b}$. 

\item For $k\in [n]$, $i\in [p]$ and $b\in \{0,1\}$, $\mathcal{S}$ chooses $\tilde{f}_{k,i,b}\in_R \textsf{H}$.

\item $\mathcal{S}$ sends $\tilde{\theta}_{k,i}$, $\tilde{f}_{k,i,b}$, $\tilde{e}_{k,i,b}\coloneqq \tilde{f}_{k,i,b}(\tilde{x}_{k,i}|_{\tilde{I}_{k,i,b}}) \oplus \tilde{s}_{k,i, b}$ where $\tilde{I}_{k,i,b}\coloneqq \{k:\tilde{\theta}_{k,i}=[+,\times]_b\}$.

\item $\mathcal{S}_{\texttt{s}}$ sends the encoded programs $\tilde{E}(\tilde{\mathcal{P}}_i)\coloneqq (\tilde{B}^b_{j,i}\tilde{D}^{-1}_{j-1,i}\tilde{M}^b_{j,i}\tilde{D}_{j,i})_{j\in [N],i\in [p]}$ to $\mathcal{A}_{\texttt{s}}$. 

\item $\mathcal{S}$ outputs the output of $\mathcal{A}_{\texttt{s}}$.
\end{enumerate} }
  }}
\smallskip

We prove that the purified version is a simulation of the adversary. The non-purified version is indistinguishable from the purified version since all the actions of the simulator commute with the adversary. This is a standard argument that was also used to prove oblivious transfer security in \cite{dfrss07}.

Let $E$ be a random state that describes $\advs$'s stored state after the qmemory bound applies. Let $\tilde{X}_{k,i}$ be the random variable describing $\mathcal{S}$'s outcome from measuring its half of the EPR pairs $|\phi_{k,i,w}\rangle$ for $w\in [m]$ in basis $\tilde{\Theta}_{k,i}$. 

The uncertainty relation, Lemma \ref{uncertinity relation}, implies that for $\gamma =\frac{1}{2m}$, there exists $\epsilon'$ negligible in $m$ such that $H_{\infty}^{\epsilon'}(\tilde{X}_{k,i}|\tilde{\Theta}_{k,i})\geq \frac{m}{2}-1$. Let $\tilde{X}_{k,i,c}\coloneqq \tilde{X}_{k,i}|_{\tilde{I}_{k,i,c}}$ for $c\in \{0,1\}$. By the Conditional Min-Entropy Splitting Lemma \ref{splitting first} and the Chain Rule \ref{chain}, there exists a random variable $\tilde{C}_{k,i}$ such that $H_{\infty}^\epsilon(\tilde{X}_{k,i,\tilde{C}_{k,i}}|\tilde{C}_{k,i}\tilde{\Theta}_{k,i}E)\geq \frac{m}{8}-\texttt{s}-2$ for $\epsilon=\epsilon'+\frac{1}{2^{m/8}}$. The variable $\tilde{C}_{k,i}$ is determined by the distribution of $\tilde{X}_{k,i,1}$ as described in the proof of the Conditional Min-Entropy Splitting Lemma \ref{splitting first} given in \cite{dfrss07}. The simulator can determine this distribution through a deep analysis of the adversary's circuit as described in Lemma \ref{determin C}. If $\Pr{[\tilde{X}_{k,i,1}=\tilde{x}_{k,i,1}]}\geq 2^{-m/4-2}$, then $\tilde{C}_{k,i}=1$ and $\tilde{C}_{k,i}=0$ otherwise. It is important to note that $\tilde{C}_{k,i}$ is completely determined once the distribution of $\tilde{X}_{k,i,1}$ is determined which in turn is determined given $\advs$ and $\tilde{\Theta}_{k,i}$. Let $\tilde{S}_{k,i,b}$ be the random variable which describes the string $\tilde{s}_{k,i,b}$. It was shown in the proof of security of oblivious transfer (Theorem \ref{QOT}) that $\tilde{S}_{k,i,1-{C}_{k,i}}$ is indistinguishable from a random string with respect to $\mathcal{A}_{\texttt{s}}$ and, hence, the simulator chooses these strings randomly. More explicitly, 
\begin{align*}\delta (\rho_{\langle \tilde{S}_{k,i,1\textsf{-}\tilde{C}_{k,i}}\tilde{S}_{k,i,\tilde{C}_{k,i}}\tilde{C}_{k,i}\tilde{E}(\tilde{\mathcal{P}}_i)\rangle_{k,i}\mathcal{A}}, \mathds{1} \otimes \rho_{\langle \tilde{S}_{k,i,\tilde{C}_{k,i}}\tilde{C}_{k,i}\tilde{E}(\tilde{\mathcal{P}}_i))\rangle_{k,i} \mathcal{A}})\leq \negl[n]. \end{align*}

$\mathcal{S}$ uses its oracle access to receive evaluations $\mathcal{P}_i(\tilde{C}_{0,i}\tilde{C}_{1,i}...\tilde{C}_{n-1,i})$ for all $i\in [p]$. Assume $\mathcal{P}_0(\tilde{C}_{0,0}...\tilde{C}_{n-1,0})=0$, then $\mathcal{S}$ chooses $\tilde{M}_{0,0}^{0}=\tilde{M}_{0,0}^{1}=Q_{\text{rej}}$ and sets all other $\tilde{M}_{j,0}^{b}$ as the identity so that the fake program evaluates to 0 as well. 

By Lemma \ref{path},
\begin{align*}
  \delta ( \rho_{\langle S_{k,i,\tilde{C}_{k,i}}\tilde{C}_{k,i} E(\mathcal{P}_i)\rangle_{k,i}\mathcal{A}},\rho_{\langle \tilde{S}_{k,i,\tilde{C}_{k,i}}\tilde{C}_{k,i}\tilde{E}(\tilde{\mathcal{P}}_i)\rangle_{k,i}\mathcal{A}})=0 .
\end{align*}

Note that the variable $\tilde{C}_{k,i}$ has the same distribution as $C_{k,i}$.

To sum up,
\begin{align*}
\rho_{{\langle {S}_{k,i,1\textsf{-}C_{k,i}}S_{k,i,{C}_{k,i}}{C}_{k,i}E(\mathcal{P}_i)\rangle_{k,i}}\mathcal{A}}&\approx_{\negl} \mathds{1} \otimes \rho_{\langle {S}_{k,i,C_{k,i}}C_{k,i}E(\mathcal{P}_i)\rangle_{k,i}\mathcal{A}} \\
&=\mathds{1} \otimes \rho_{\langle \tilde{S}_{k,i,C_{k,i}}C_{k,i}\tilde{E}(\tilde{\mathcal{P}}_i)\rangle_{k,i}\mathcal{A}}\\
&\approx_{\negl} \rho_{\langle \tilde{S}_{k,i,1\textsf{-}C_{k,i}}\tilde{S}_{k,i,{C}_{k,i}}{C}_{k,i}\tilde{E}(\tilde{\mathcal{P}}_i)\rangle_{k,i}\mathcal{A}} \\
&\approx_{\negl} \rho_{\langle \tilde{S}_{k,i,1\textsf{-}\tilde{C}_{k,i}}\tilde{S}_{k,i,{\tilde{C}}_{k,i}}{\tilde{C}}_{k,i}\tilde{E}(\tilde{\mathcal{P}}_i)\rangle_{k,i}\mathcal{A}} .
\end{align*}

The first state belongs to the adversary while the final state is generated by the simulator. 

\qed     \end{proof}

Note that security holds independent of any auxiliary input provided to the simulator and adversary. The security is built on the adversary's ignorance of the matrix sets $\{D_{i}\}_{i\in[N]}$ and $\{B_{i}^b\}_{i\in [N]}$ used in the construction which cannot be revealed in any auxiliary inputs since these sets are part of the randomness used in the construction. 

Furthermore, the one-time programs also satisfy a form of receiver security which will be useful in proving IND-CCA1 encryption security. This is a direct result of receiver security of oblivious transfer (Theorem \ref{QOT}). 

\begin{theorem}
\label{receiver security}
For any dishonest sender and honest evaluator $\mathcal{E}$ following the protocol described in Construction \ref{con:branching program}, there exists a branching program $P'$ such that $\mathcal{E}$ obtains $P'(w)$ for any input $w$ chosen by $\mathcal{E}$. 
\end{theorem}

\subsection{Extending to \texorpdfstring{$\textbf{NC}^1$}{NC1} circuits}
\label{sec:nc}

Let $C$ be a circuit in $\textbf{NC}^1$ of depth $d$ with $n_0$ inputs and $n_1$ outputs. First, $C$ is split into $n_1$ circuits with 1-bit outputs. Each circuit is then transformed into a branching program $P_{\ell}=(M_{j,\ell}^b)_{j\in n_04^d}$ using Barrington's Theorem \cite{B89}. Note that the length of the branching program depends on the circuit, however, any $\textbf{NC}^1$ circuit of depth $d$ can be transformed into a branching program of length $n_04^d$ (see Sec.~\ref{sec:branching}). We pad the program with identity instructions if necessary to ensure the length is $n_04^d$ as otherwise the length of the program would reveal information regarding the circuit.

The one-time compiler follows the same procedure except with a slight modification. Like before, the protocol randomly chooses $(B_{j,l}^b)_{j\in n_04^d,l\in [n_1]}$ and $(D_{j,\ell})_{j\in n_04^d-1,\ell\in [n_1]}$ and encodes the programs as:
\begin{align*}
  E(P_\ell)= (B^b_{j,\ell}D^{-1}_{j-1,\ell}M_{j,l}^{b}D_{j,\ell})_{j\in [n_04^d],b\in \{0,1\}} .
\end{align*}
However, these programs are decoded jointly: each string $s_{k}^b$ reveals $B^b_{j,\ell}$ for all $j=k\bmod n_0$ and $ \ell \in [n_1]$. Hence, all the programs can only be evaluated on the same input. However, this means the strings $s_k^b$ are longer so a larger $m$ value in the OT transmissions $\rho_{m,\ell}(s_{k}^0,s_{k}^1)$ is required to maintain security (Theorem \ref{QOT}). The value $m= 250n_04^dn_1+4n_1\texttt{s}$ turns out to be sufficient. 

\subsection{Extending to Polynomial Circuits}
\label{sec:poly obfuscation}

If $C$ is a circuit of size $N\in \poly$ and depth $d$ then $\mathcal{O}_{\texttt{s}}(C)$ is obtained as follows. First $C$ is sliced into smaller circuits $C_0,C_1,...,C_{M-1}$ where $M\coloneqq \lceil \frac{d}{2\lg N}\rceil$ and $C_i\in \textbf{NC}^1$ for all $i\in [M]$. This is done by taking levels 1 to $2\lg N$ as a subcircuit which we denote as $C_0$ and levels $2\lg N$ to $4\lg N$ as $C_1$, and so forth. Then for each subcircuit, $C_i$, the input and output length is expanded to length $N$ and padded with a randomly chosen value $k_i\in_R\{0,1\}^N$. Correspondingly, at the beginning of $C_i$ a step is added to decrypt the one-time padded output of $C_{i-1}$ using $k_{i-1}$. The resulting circuits are denoted as $\hat{C}_{i,k_{i-1},k_i}$. Next, each $\hat{C}_{i,k_{i-1},k_i}$ is converted to a one-time program as described in Sec.~\ref{sec:nc}. 

To allow a user to evaluate $C$, a decryption circuit is used which reveals the final key $k_{M-1}$. The decryption circuit, denoted as $D_{k_{M-1}}$, takes as input the outputs of all the subcircuits $(\hat{C}_{i,k_{i-1},k_i})_{i\in [M]}$, verifies the evaluation is performed correctly and if so returns $k_{M-1}$ and if not returns $\perp$. Hence, the evaluator is required to keep track of the outputs of all the subcircuits which are used as proof that the evaluation was performed correctly i.e. the output $\hat{C}_{i-1,k_{i-2},k_{i-1}}$ is fed as input to $\hat{C}_{i,k_{i-1},k_i}$ and so on. The verification step can be done in $\textbf{NC}^1$ since all the input/output pairs for all the subcircuits are provided by the evaluator so the decryption circuit just needs to check these values. The key $k_{M-1}$ allows the evaluator to obtain the evaluation of $C$ on a single input. Note that the purpose of the decryption circuit is to reveal the key only to users which perform the calculation correctly. Hence, the key must be kept hidden which is why a one-time program of $D_{k_{M-1}}$ is sent to the evaluator instead of the actual circuit. 

\begin{construct}[One-Time Compiler for Polynomial Circuits]
\label{con:polynomial one-time}
{\small Let $C$ be a circuit of size $N\in \poly$ and depth $d$ and let $C_0$, $C_2$,..., $C_{M-1}$ be the subcircuits of $C$ of depth $2\lg N$ where $M\coloneqq \lceil \frac{d}{2\lg N}\rceil$. 
\begin{enumerate}
\item Add extra wires to extend the input and output lengths of each subcircuit to $N$.

\item Pick random binary strings $(k_i)_{i\in [M]}$ where $\lvert k_i\rvert =N$.

\item Add to the end of $C_i$ an encryption step that adds the output to $k_i$ and add at the start a decryption step that decrypts the input by adding $k_{i-1}$. Denote the resulting circuit as $\hat{C}_{i,k_{i-1},k_i}$

\item Send one-time programs $(\mathcal{O}_{\texttt{s}}(\hat{C}_{i,k_{i-1},k_i}))_{i\in [M]}$ one at a time in sequence. 

\item Construct a decryption circuit $D_{k_{M-1}}$ that takes as input $\langle (a_i,b_i)\rangle _{i\in M}$ where $a_i,b_i\in \{0,1\}^N$ except $a_0\in \{0,1\}^n$ and checks:
\begin{itemize}
  \item $\hat{C}_{i,k_{i-1},k_i}(a_i)=b_i$ for all $i\in [M]$.
  \item $b_i=a_{i+1}$ for all $i\in [M]$.
\end{itemize}
If the checks are passed, $D_{k_{M-1}}$ outputs $k_{M-1}$ and $\perp$ otherwise. 
Send $\mathcal{O}_{\texttt{s}}(D_{k_{M-1}})$. 
\end{enumerate} }\end{construct}

The idea behind the security of this protocol is that each one-time program only reveals one output but this value is one-time padded. An evaluator can only learn the pad $k_{M-1}$ of the final subcircuit using the decryption circuit by performing the evaluations correctly. 

\begin{theorem}[One-time Polynomial Circuit Security]
\label{polynomial circuit}
Let $n, \texttt{s} \in \mathbb{N}$ and let $(C)_{i\in [p]}$ be a set of circuits with domain $\{0,1\}^n$ and lengths $\{N_i\}_{i\in [p]}$ polynomial in $n$. Then for any adversary $\adv_{\texttt{s}}$, there exists a simulator $\mathcal{S}_{\texttt{s}}$ such that 
\[
\begin{split}
|Pr[\adv(\mathcal{O}_{\texttt{s}}(C_0),...,\mathcal{O}_{\texttt{s}}(C_{p-1}))=1]-Pr[ \mathcal{S}_{\texttt{s}}^{1\mathcal{C}_0,...,1\mathcal{C}_{p-1}}(|0\rangle^{\otimes N_0},...,|0\rangle^{\otimes N_{p-1}})=1]\rvert & \\
\leq \negl[n] .&
\end{split}
\]
\end{theorem}

A similar and detailed proof can be found in \cite{k88}.
\begin{proof}[sketch]
For simplicity, we only prove the case where the adversary has access to a single one-time program of a circuit $C$ of size $N$ and depth $d$. The proof can be easily adapted to the more general case. 

\smallskip \noindent\fbox{%
  \parbox{\textwidth}{%
\textbf{Simulator $\mathcal{S}_{\texttt{s}}^{1P}(|0\rangle^{\otimes N})$}

{\small Let $M\coloneqq \lceil \frac{d}{2\lg N}\rceil$.
\begin{enumerate}
\item $\mathcal{A}_{\texttt{s}}$ is run. 

\item $\mathcal{A}_{\texttt{s}}$ requests a copy of the one-time program of $C$. 

\item For $i\in [M]$, $\mathcal{S}_{\texttt{s}}$ picks $a_{i},b_{i}\in_R \{0,1\}^N$. Let $\tilde{C}_i$ be a $\mathbf{NC}^1$ circuit of depth $2\lg N$ mapping $x\in \{0,1\}^N$ to $a_ix+b_i \mod 2^N$.

\item For $i\in [M]$, $\mathcal{S}_{\texttt{s}}$ sends one-time programs $\mathcal{O}_{\texttt{s}}(\tilde{C}_i)$. 

\item $\mathcal{S}_{\texttt{s}}$ analyses the action of $\mathcal{A}_{\texttt{s}}$ on $\mathcal{O}_{\texttt{s}}(\tilde{C}_0)$ and determines $\tilde{w}_0$ such that $\advs$'s access to $\mathcal{O}(\tilde{C}_0)$ can be simulated with the evaluation $\tilde{C}_0(\tilde{w}_0)$.

\item For each $i\in [M]$, $\mathcal{S}_{\texttt{s}}$ calculates $\tilde{w}_i\coloneqq \tilde{C}_{i-1}(\tilde{w}_{i-1})$. 
 
\item $\mathcal{S}_{\texttt{s}}$ queries the oracle to learn $C(\tilde{w}_0)$ and constructs a decryption circuit $\tilde{D}_{\tilde{k}_{M-1}}$ which takes as input $(\langle x_i,y_i\rangle)_{i\in [M]}$ and checks if $x_i=\tilde{w}_i$ and $y_i=\tilde{w}_{i+1}$ and if so returns $\tilde{k}_{M-1}$ where $\tilde{k}_{M-1}$ is chosen to satisfy $\tilde{C}_{M-1}(\tilde{w}_{M-1})\oplus \tilde{k}_{M-1}=C(\tilde{w}_0)$. If these conditions are not satisfied then the decryption circuit outputs $\perp$. $\mathcal{S}_{\texttt{s}}$ sends $\mathcal{O}_{\texttt{s}}(\tilde{D}_{\tilde{k}_{M-1}})$ to $\mathcal{A}_{\texttt{s}}$. 

\item $\mathcal{S}_{\texttt{s}}$ outputs the output of $\mathcal{A}_{\texttt{s}}$.
\end{enumerate} }}}
  \smallskip
  
For each $i\in [M]$, the one-time program $\mathcal{O}_{\texttt{s}}(\hat{C}_{i,k_{i-1},k_i})$ can be simulated with a single evaluation, say $\hat{C}_{i,k_{i-1},k_i}(w_i)$, by Theorem \ref{main}. 

There are two cases to consider. First, if the evaluation was not performed correctly: for some $i\in [M]$, $w_{i+1}\neq {C}_{i,k_{i-1},k_i}(w_i)$. Then except with negligible probability the one-time program $\mathcal{O}(D_{k_{M-1}})$ can be simulated with $\perp$. In summary, the adversary can be simulated with the values $( {C}_{i,k_{i-1},k_i}(w_i))_{i\in [M]}$ and $\perp$. All these values are random since the keys $(k_i)_{i\in [M]}$ are never revealed so the simulation is easy and does not even require an oracle query to $C$.

In the second case, the evaluation is performed correctly: for all $i\in [M]$, $w_{i+1}= \hat{C}_{i,k_{i-1},k_{i}}(w_i)$. In this case, the program $\mathcal{O}(D_{k_{M-1}})$ can be simulated with $k_{M-1}$. Overall, the adversary can be simulated with the values $(w_i)_{i\in [M+1]}$ and $k_{M-1}$. These values are all random with the exception of $w_M$ which satisfies the relation $w_M\oplus k_{M-1}=C(w_0)$. Hence, in either case, the simulation requires at most one oracle query to $C$.

The details of the simulation are in the box above. Roughly speaking, the values $(\tilde{w}_i)_{i\in [M]}$ and $\tilde{k}_{M-1}$ are random and hence simulate $({w}_i)_{i\in [M]}$ and $k_{M-1}$, respectively. Moreover, $\tilde{w}_{M}\oplus \tilde{k}_{M-1}=C(\tilde{w}_0)$ so the relation is satisfied and $\tilde{w}_M$ simulates $w_M$. 

\qed     \end{proof}

The $1\textsf{-}2$ oblivious transfer sub-protocols in our one-time program are disappearing by Theorem \ref{disappearing OT}, giving the following result. 

\begin{theorem}
$\mathcal{O}_{\texttt{s}}$ is a disappearing information-theoretically secure one-time compiler for the class of polynomial classical circuits against any computationally unbounded adversary with $\texttt{s}$ qmemory bound. 
\end{theorem}

\section{Encryption Tokens}
\label{sec:enc tok}

An encryption token allows its holder to encrypt only a single message while a decryption token allows its holder to decrypt only a single ciphertext. We present an information-theoretic secure scheme for encryption and decryption tokens in the \textsf{BQSM}. Tokens disappear after they are received. 

\subsection{Definition}
\label{sec:enc tok def}
We formally define the notions of encryption and decryption tokens. 

\begin{definition}[Encryption and Decryption Tokens] 
A $q$-time encryption token scheme over message space $\hil{M}$ consists of the following QPT algorithms:

\begin{itemize}
  \item $\textsf{Gen}(1^\lambda,q):$ Outputs a classical key $k$. 
  \item $\textsf{EncTok}(q, k):$ Outputs a state $\rho_e$ using the key $k$ (at most $q$ times). 
  \item $\textsf{Enc}(\rho_{e},\mu):$ Outputs a ciphertext $ct$ for $\mu$ using $\rho_e$. 
  \item $\textsf{DecToken}(k):$ Outputs a state $\rho_d$ using the key $k$.
  \item $\textsf{Dec}(\rho_d,ct):$ Outputs a message $\mu'$ for ciphertext $ct$ using $\rho_d$. 
\end{itemize}
\end{definition}

\begin{definition}[Correctness]
A $q$-time token scheme is \emph{correct} if for any message $\mu\in \hil{M}$,
\begin{align*} \Pr{
\left[\begin{tabular}{c|c}
 \multirow{4}{*}{$\textsf{Dec}(\rho_d,ct)=\mu\ $} & $k\ \leftarrow \textsf{Gen}(1^\lambda,q)$ \\ 
 & $\rho_e \leftarrow \textsf{EncTok}(q, k)$\\
 & $\rho_d \leftarrow \textsf{DecToken}( k)$\\
 & $ct\ \leftarrow \textsf{Enc}(\rho_e,\mu)$
 \end{tabular}\right]} \geq 1-\negl[\lambda] .
\end{align*}
\end{definition}

Token security requires that $e$ encryption and $d$ decryption tokens where $d+e< q$ can only be used to generate $d+e$ valid message-ciphertext pairs. 

\smallskip \noindent\fbox{%
  \parbox{\textwidth}{%
\textbf{Experiment} $\textsf{TokExp}_{\Pi,\mathcal{A}}(\lambda, q,e,d):$
\begin{enumerate}
  \item Sample a key $k\leftarrow \textsf{Gen}(1^\lambda,q)$. 
  \item $(\mu_i)_{i\in [d+2]}\leftarrow \adv$.
  \item Choose a random subset of $d+1$ messages $(\mu_{j_i})_{i\in [d+1]}$ from $(\mu_i)_{i\in [d+2]}$. 
  \item Generate $d+1$ encryption tokens and encrypt the chosen messages. 
  \item Send $\adv$ the resulting ciphertexts.
 \item $(m_i,ct_i)_{i\in [d+e+1]}\leftarrow \adv^{e\textsf{EncToken},d\textsf{DecToken}}$. 
 \item Generate $d+e+1$ decryption tokens $(\rho^i_{d})_{i\in [d+e+1]}$
  \item The output of the experiment is $1$ if all $m_i$ are distinct and $\textsf{Dec}(\rho^i_{d},ct_i)=m_i$ for all $i\in [d+e+1]$ and 0 otherwise.
\end{enumerate}
}}\smallskip

\begin{definition}[Security]
A $q$-time encryption token scheme $\Pi_{\texttt{s}}$ is \emph{secure} in the \textsf{BQSM} if for any adversary $\advs$ and $e,d\in \mathbb{N}$ such that $e+d< q$,
\begin{align*}
  \Pr{[\textsf{TokExp}_{\Pi_{\texttt{s}},\advs}(\lambda, q,e,d)=1]}\leq \frac{1}{2}+ \negl[\lambda] .
\end{align*}
\end{definition}

Note that tokens are very similar to one-time programs since an encryption token can be used to obtain the encryption of a single message chosen by the user. Hence, we can define disappearing security in the same way as for one-time programs (Sec.~\ref{def dis otp}).

\subsection{Construction}

We first construct a token scheme that allows for the distribution of a limited number of encryption tokens and unbounded polynomial number of decryption tokens. Then, we extend our construction to allow for an unbounded polynomial number of encryption tokens by using a program broadcast.

We say that a $q$-time encryption token scheme is secure if an adversary that is given $e$ encryption tokens and $d$ decryption tokens, where $e+d<q$, cannot produce more than $d+e$ valid message-ciphertext pairs. See App.~\ref{sec:enc tok def} for the formal definition.

The backbone of the construction is a $q$-time pad. The encryption and decryption tokens are simply one-time programs of the encryption and decryption functionalities in the underlying $q$-time pad scheme. 

\begin{construct}
\label{te}
{\small
The $q$-time encryption token scheme $\Pi_{\textsf{EncTok}}$ against adversaries with qmemory bound $\texttt{s}$ over message space $\hil{M}=\{0,1\}^{\ell}$ with $q$ encryption tokens is as follows.}

\begin{itemize}
  \item $\textsf{Gen}(1^\lambda,q)$: Let $n=\max( \lambda,q)+1$. Choose a $(6n)\times (3n)$ binary matrix $M$ that is left invertible uniformly at random. Let $M^{-1}$ be a left inverse of $M$. Choose a polynomial $F$ of degree $n$ over $\{0,1\}^{6n}$ uniformly at random. The key is $k=(M, M^{-1},F)$. 

  \item $\textsf{EncTok}(\ell, q,\texttt{s},k )$: 
  Pick a string $r\in_R \{0,1\}^{2n}$. Construct a program $\textsf{Enc}_{k,r}$ which takes as input $x \in\{0,1\}^\ell$ and outputs $(M\cdot {(r\|x)},F(M\cdot {(r\|x)}))$. Send the state $\rho_e\leftarrow \mathcal{O}_{\texttt{s}}(\textsf{Enc}_{k,r})$ (at most $q$ times).
  
  \item $\textsf{Enc}(\rho_e,\mu):$ Evaluate the one-time program $\rho_e$ on $\mu$ and return the output.
  
  \item $\textsf{DecToken}(\ell, \texttt{s}, k):$ 
  Construct a program $\textsf{Dec}_k$ which takes as input $x=(x_0,x_1)\in \{0,1\}^{6n}\times \{0,1\}^{6n}$ and checks if $x_1=F(x_0)$ and if so outputs the last $\ell$ bits of $M^{-1}\cdot {x}_0$ and outputs $\perp$ otherwise. Send the state $\rho_d\leftarrow \mathcal{O}_{\texttt{s}}(\textsf{Dec}_k)$. 
 
  \item $\textsf{Dec}(\rho_d, ct)$: Evaluate the one-time program $\rho_d$ on $ct$ and return the output.
\end{itemize}
\end{construct}

\begin{theorem}[Security]
\label{enctok security}
Construction \ref{te} ($\Pi_{\textsf{EncTok}}$) is secure against computationally unbounded adversaries with $\texttt{s}$ qmemory bound. Moreover, tokens are disappearing against the same class of adversaries. 
\end{theorem}

\begin{proof}
An adversary $\mathcal{A}_{\texttt{s}}$ in the $\textsf{TokExp}_{\Pi,\mathcal{A}}(\lambda, q,e,d)$ experiment has access to $e$ queries to $\textsf{EncTok}(\ell, q,\texttt{s},k )$, $d$ decryption queries $\textsf{DecToken}(\ell,\texttt{s}, k)$ and $d+1$ ciphertexts $((ct_i,ct_i'))_{i\in [d+1]}$ where $(ct_i,ct'_i)\coloneqq \textsf{Enc}_{k,r_i}(\mu_{j_i})$ and $ct'_i=F(ct_i)$ for some random string $r_i\in \{0,1\}^{2n}$. By the security of one-time programs (Theorem \ref{main poly}), such an adversary can be simulated by an algorithm $\mathcal{S}_{\texttt{s}}$ with access to a single oracle query to $\textsf{Enc}_{k,r_i}$ where $r_i\in \{0,1\}^{2n}$ is a random string for each $i\in [d+1,d+1+e]$ and $d$ queries to $\textsf{Dec}_k$. From the queries to the encryption oracle, $\mathcal{S}_{\texttt{s}}$ obtains ciphertexts $((ct_i,ct'_i))_{i\in [d+1,d+e+1]}$ where $(ct_i,ct'_i)\coloneqq Enc_{k,r_i}(\mu_i)$ for some message $\mu_i$. By Theorem \ref{lagrange 2}, since $d+e+1\leq q<n$ there is negligible probability $\mathcal{S}_{\texttt{s}}$ can guess an output $F(v)$ of a vector $v$ which is not in the set $(ct_i)_{i\in [d+e+1]}$. This means the ciphertexts submitted to the oracle of $\textsf{Dec}_k$ that are not in the set $((ct_i,ct_i'))_{i\in [d+e+1]}$ will return $\perp$ except with negligible probability. 

Without loss of generality, assume the simulator uses the oracle to decrypt the ciphertexts $((ct_i,ct_i'))_{i\in [d]}$ (it is useless to decrypt any ciphertext in $((ct_i,ct_i'))_{i\in [d+1,d+e+1]}$). After the oracle queries have been exhausted, the simulator has learned $d+e$ message-ciphertext pairs and it remains to guess one more pair to break security. As mentioned, there is negligible probability $\mathcal{S}_{\texttt{s}}$ can guess an output $F(v)$ of a vector $v$ which is not in the set $(ct_i)_{i\in [d+e+1]}$. This implies that $\mathcal{S}_{\texttt{s}}$ must decrypt the ciphertext $(ct_{d+1},ct'_{d+1})$. There remains two possible values for the hidden message $\mu_{j_{d+1}}$ say $\mu_0$ and $\mu_1$. Let the correct answer be $\mu_b$. There is negligible probability that $r_{d+1}$ is in the span of $(r_i)_{i\in [d+e+1]}/r_{d+1}$ so part 2 of Theorem \ref{Raz 2} implies that if $\mathcal{S}_{\texttt{s}}$ outputs a guess $b'$ then $\Pr{[b'=b]}\leq \frac{1}{2}+O(2^{-\lambda})$.

The encryption and decryption tokens are one-time programs so they are disappearing by Theorem \ref{main poly}.

\qed     \end{proof}

\subsection{Sparse and Ordered Tokens}
We now discuss how to modify our token scheme to allow for the distribution of an unbounded polynomial number of tokens. The idea is to simply perform a $q$-time program broadcast of the encryption and decryption tokens. More explicitly, instead of sending $\mathcal{O}_{\texttt{s}}(\textsf{Enc}_{k,\cdot})$ whenever an encryption token is requested, we perform a $q$-time program broadcast of the function $\textsf{Enc}_{k,\cdot}$. Also, instead of sending $\mathcal{O}_{\texttt{s}}(\textsf{Dec}_k)$, we broadcast $\textsf{Dec}_k$. 

Our program broadcast protocol ensures that no adversary can obtain more than $q$ encryptions regardless of the (polynomial) number of tokens distributed (Theorem \ref{broadcast security}). Security then follows in the same way as in the proof of Theorem \ref{enctok security}. 

The disadvantage of this approach is that a broadcast only lasts for a set time (see Construction \ref{con:broadcaster}). Hence, to allow for continual use, the private key needs to grow (linearly) with time or, more precisely, a new key needs to be generated for each broadcast. 

This modification naturally imposes a sparsity restraint on the tokens as there is a limit to the number of decryptions or encryptions a user can perform in a set time. With each token broadcast, a user can perform at most $q$ encryptions or decryptions. 

Another restraint that can be easily implemented is to enforce encryptions or decryptions to be performed in a specific order. This can be incorporated into the encryption or decryption token functionalities. In particular, the function $\textsf{Enc}_{k,\cdot}$ broadcasted can take additional arguments to check if the user has performed the required encryptions or decryptions prior to revealing the encryption of a message further up the order. As an interesting application of the time and order restraints used in conjunction, we can create timed encryption token. This is a token scheme where the authority can control when each ciphertext becomes decryptable. Similar notions of sparsity and order restraints were studied by Amos, Georgiou, Kiayias, and Zhandry \cite{AGK20} and by Georgiou and Zhandry \cite{GZ20} in the plain model. However, their constructions rely on heavy computational assumptions such as indistinguishability obfuscation or the existence of classical oracles, whereas our constructions are information-theoretically secure in the \textsf{BQSM}.

\section{Signature Tokens}

A signature token allows its holder to sign only a single message while a verification token allows its holder to verify only a single signature. We present a scheme for signature and verification tokens in the \textsf{BQSM} which disappear after they are received. 

\subsection{Definition}

We introduce signature tokens similar to \cite{BSS16}. 

\begin{definition}[Signature Tokens]
A \emph{$q$-time signature token scheme} $\Pi$ over message space $\hil{M}$ consists of the following QPT algorithms:
\begin{itemize}
  \item $\textsf{Gen}(1^\lambda,q):$ Outputs a classical key $k$.
  \item $\textsf{SignToken}(q, k):$ Sends a state $\rho_s$ at most $q$ times.
  \item $\textsf{Sign}(\mu, \rho_s):$ Outputs a signature $\sigma$ of $\mu\in \hil{M}$ using $\rho_s$.
  \item $\textsf{VerToken}( k):$ Sends a state $\rho_v$.
  \item $\textsf{Verify}(\rho_v , (\mu*, \sigma*)):$ Verifies $\sigma*$ is a signature of $\mu*$ using $\rho_v$ and correspondingly output a bit $b$. 
\end{itemize}
\end{definition}

\begin{definition}[Correctness]
We say a signature token scheme is \emph{correct} if for any message $\mu\in \hil{M}$,
\begin{align*} \Pr{
\left[\begin{tabular}{c|c}
 \multirow{4}{*}{$\textsf{Verify}(\rho_v,(\mu, \sigma))=1\ $} & $k\ \leftarrow \textsf{Gen}(1^\lambda,q)$ \\ 
 & $\rho_s \leftarrow \textsf{SignToken}(q, k)$\\
 & $\sigma\leftarrow \textsf{Sign}(\rho_s, \mu)$\\
 & $\rho_v \leftarrow \textsf{VerToken}( k)$\\
 \end{tabular}\right]} \geq 1-\negl[\lambda] .
\end{align*}
\end{definition}

\smallskip \noindent\fbox{%
  \parbox{\textwidth}{%
\textbf{Experiment} $\textsf{SignToken}_{\Pi,\mathcal{A}}(\lambda, q):$
\begin{enumerate}
  \item Sample a key $k\leftarrow \textsf{Gen}(1^\lambda,q)$. 
   \item $(\mu_i, \sigma_i)_{i\in [q+1]}\leftarrow \adv^{q\textsf{SignToken},\textsf{VerToken}}$.
   \item Generate $q+1$ verification tokens $(\rho_v^i)_{i\in [q+1]} \leftarrow \textsf{VerToken}(k)$. 
  \item The output of the experiment is 1 if all $\mu_i$ are distinct and $\textsf{Verify}(\rho_v^i,\mu_i, \sigma_i)=1$ for all $i\in [q+1]$ and 0 otherwise. 
\end{enumerate}}}
\smallskip

\begin{definition}[Security]
A $q$-time signature token scheme $\Pi_{\texttt{s}}$ is secure in the \textsf{BQSM} if for any adversary $\advs$,
\begin{align*}
  Pr[\textsf{SignToken}_{\Pi_{\texttt{s}},\advs}(\lambda, q)=1]\leq \negl[\lambda] .
\end{align*}
\end{definition}

Note that these tokens are very similar to one-time programs so we can define disappearing security in the same way as for one-time programs (Sec.~\ref{def dis otp}).

\subsection{Construction}

We now present a $q$-time signature token scheme in the \textsf{BQSM}. The backbone of the construction is a $q$-time \textsf{MAC} scheme. The signature and verification tokens are simply one-time programs of the signature and verification functionalities in the underlying $q$-time \textsf{MAC} scheme. 

\begin{construct}
\label{sigtok}
{\small The signature token scheme $\Pi_{\textsf{SigTok}}$ against adversaries with $\texttt{s}$ qmemory bound over message space of $\hil{M}=\{0,1\}^\ell$ consists of the following algorithms:

\begin{itemize}
  \item $\textsf{Gen}(1^\lambda,q)$: Let $n=\max(\lambda ,q)+1$. Choose a polynomial $F$ of degree $n-1$ over the field of order $2^{n}$ uniformly at random. The key is $k=F$. 
  
  \item $\textsf{SignToken}(q,\texttt{s},k )$: 
  Send the state $\rho_s\leftarrow \mathcal{O}_{\texttt{s}}(F)$ (at most $q$ times).
  
  \item $\textsf{Sign}(\rho_s,\mu):$ Evaluate the one-time program $\rho_s$ on $\mu$ and return the output.
  
  \item $\textsf{VerToken}(\texttt{s}, k):$ 
  Let $V$ be a program which takes two inputs $(\mu, \sigma)$ and outputs 1 if $F(\mu)=\sigma$ and outputs $\perp$ otherwise. Send the state $\rho_v \leftarrow \mathcal{O}_{\texttt{s}}(V)$. 
 
  \item $\textsf{Verify}(\rho_v, (\mu^*,\sigma^*))$: Evaluate the one-time program $\rho_v$ on $(\mu^*,\sigma^*)$ and return the output.
\end{itemize}}
\end{construct}

\begin{theorem}[Security]
\label{sig token security}
Construction \ref{sigtok} ($\Pi_{\textsf{SigTok}}$) is a secure token scheme against computationally unbounded adversaries with $\texttt{s}$ qmemory bound. Moreover, tokens are disappearing against the same class of adversaries. 
\end{theorem}

\begin{proof}
This is a direct consequence of Theorem \ref{lagrange 2} and one-time program security (Theorem \ref{main poly}). 

\qed     \end{proof}

\subsection{Sparse and Ordered Tokens}
\label{sec:SOsig}

Similar to App.~\ref{sec:enc tok}, we can modify Construction \ref{sigtok} ($\Pi_{\textsf{SigTok}}(\lambda,\ell, q,\texttt{s})$) to allow for unbounded polynomial number of tokens. Instead of sending $\mathcal{O}_{\texttt{s}}(F)$ whenever a signature token is requested, we perform a $q$-time program broadcast of $F$. Also, instead of sending $\mathcal{O}_{\texttt{s}}(V)$ whenever a verification token is requested, we broadcast $V$. 

Our program broadcast protocol ensures that no adversary can obtain more than $q$ signatures regardless of the (polynomial) number of tokens distributed (Theorem \ref{broadcast security}). Security then follows in the same way as in the proof of Theorem \ref{sig token security}. 

Again similar to the encryption token setting, the private key needs to grow (linearly) with time. Moreover, sparsity and order restraints can be imposed in the same way.

\section{Proofs}

\subsection{Proof of Theorem \ref{Raz 2}}
\label{app:Raz 2}

\begin{proof}
For simplicity assume that the vectors $\{a_i\}_i$ are linearly independent. In the case they are dependent then $\Pr{[b'=b]}$ is smaller since $A$ is given less information regarding $M$. 

By Lemma \ref{Raz 1}, there are $ 2^{(n-m)}$ binary vectors ${v}$ such that ${v}\cdot {a}_i=(b_i)_j$ for any $i\in [m], j\in [\ell]$ and there are $2^{(n-m-1)}$ vectors which additionally satisfy ${v}\cdot {a}=b_j$. This means there are $2^{\ell(n-m)}$ matrices $M'$ which satisfy $M'\cdot a_i=b_i$ and $2^{\ell(n-m-1)}$ matrices which additionally satisfy $M'\cdot a=b$. The restriction $M'\cdot \hat{a}_1\neq \hat{b}_1$ reduces the number of matrices by at most $2^{\ell(n-m-1)}$, hence, all the restrictions $M'\cdot {\hat{a}}_i\neq \hat{b}_i$, $i\in [p]$ reduce the number of matrices by at most $p2^{\ell(n-m-1)}$. In summary, the total number of matrices which satisfy all the conditions is at least\begin{align*}
2^{\ell(n-m)}-p2^{\ell(n-m-1)}=2^{\ell(n-m)}\left(1-\frac{p}{2^{\ell}}\right) .\end{align*} 
The probability that $b'=b$ is equal to the probability $A$ guesses a matrix $M'$ such that $M'a=b$ which is less than:
\begin{align*}
  \Pr{[b'=b]}\leq \frac{2^{\ell(n-m-1)}}{2^{\ell(n-m)}\left(1-\frac{p}{2^{\ell}}\right)}=O(2^{-\ell}) .
\end{align*}
For part $(2)$, by the same reasoning as above, there are at most $2^{\ell(n-m-1)}$ matrices $M'$ that satisfy all the conditions as well as $y_r=M'\cdot {x}_r$. While there are at least $2^{\ell(n-m-1)+1}(1-\frac{p}{2^{\ell}})$ matrices $M'$ which satisfy all conditions as well as $y_r=M'\cdot {x}_{r\oplus 1}$ or $y_r=M'\cdot {x}_{r}$. Hence, the probability that $r'=r$ is equal to the probability $A$ guesses a matrix $\hat{M }$ such that $y_r=M'\cdot {x}_r$ which is less than:
\[
  \Pr{[r'=r]}\leq \frac{2^{\ell(n-m-1)}}{2^{\ell(n-m-1)+1}(1-\frac{p}{2^{\ell}})}\leq \frac{1}{2}+O(2^{-\ell})
  .
\]

\qed     \end{proof}

\subsection{Proof of Lemma \ref{lagrange}}
\label{app:lagrange}
\begin{proof}
Choose $a_{m},...,a_d\in \mathbb{F}$ and $b_m,...,b_d\in \mathbb{F}$ such that $\{a_i\}_{i\in [d+1]}$ are all distinct. For any set of pairs $(a_0,b_0),...,(a_d,b_d)$, there exists an unique polynomial $f$ of degree $d$ such that $f(a_i)=b_i$ for all $i\in [d+1]$. To show uniqueness, assume there exists two polynomials $f_1,f_2$ satisfying these conditions then there exists a non-zero polynomial of degree $\leq d$ with $d+1$ roots which is impossible. To show existence, consider the following polynomial. 
\begin{align*}
  f(x)=\sum_{i=0}^d b_i\prod_{j\neq i} \frac{x-a_j}{a_i-a_j} .
\end{align*}
Hence, the number of polynomials that satisfy $f(a_i)=b_i$ for $i\in [m]$ is $2^{\ell(d-m)}$ as a result of the possible values a polynomial can take on $a_{m},...,a_d$.
 
\qed     \end{proof}

\subsection{Proof of Theorem \ref{lagrange 2}}
\label{app:lagrange 2}
\begin{proof}
For part $(1)$, by Lemma \ref{lagrange}, there are $ 2^{(d-m)\ell}$ polynomials $f$ such that $f(a_i)=b_i$ for all $i\in [m+1]$ and there are $2^{(d-m-1)\ell}$ polynomials which additionally satisfy $f(a)=b$. The restrictions $f( \hat{a}_i)= \hat{b}_i$ reduces the number of possible polynomials by at most $p2^{\ell(d-m-1)}$. 
The probability that $b'=b$ is equal to the probability $A$ guesses a polynomial $f'$ such that $f'(a)=b$ which is less than:
\[
  \Pr{[b'=b]}\leq \frac{2^{\ell(d-m-1)}}{2^{\ell(d-m)}\left(1-\frac{p}{2^{\ell}}\right)}=O(2^{-\ell}) .
\]
For part $(2)$, by the same reasoning as above, there are at most $2^{\ell(d-m-1)}$ polynomials $f'$ that satisfy all the conditions as well as $y_r=f'(x_r)$. While there are at least $2^{\ell(d-m-1)+1}(1-\frac{p}{2^{\ell}})$ polynomials $f'$ which satisfy all conditions as well as $y_r=f'(x_{r\oplus 1})$ or $y_r=f'(x_{r})$. Hence, the probability that $r'=r$ is equal to the probability $A$ guesses a polynomial $\hat{f }$ such that $y_r=f'(x_r)$ which is less than:
\begin{align*}
  \Pr{[r'=r]}\leq \frac{2^{\ell(d-m-1)}}{2^{\ell(d-m-1)+1}(1-\frac{p}{2^{\ell}})}\leq \frac{1}{2}+O(2^{-\ell})
  .
\end{align*}

\qed     \end{proof}

\end{document}